\documentclass[12pt,transaction, draftclsnofoot,onecolumn]{IEEEtran}
\usepackage{amsfonts}
\usepackage{amssymb}
\usepackage{url}
\usepackage{graphicx}
\usepackage{subfigure}
\usepackage{enumerate}
\usepackage{amsmath}
\usepackage{color}
\usepackage{amsthm}
\usepackage{amsmath}
\usepackage{algorithm}
\usepackage{algpseudocode}
\usepackage{color}
\usepackage{accents}
\usepackage{lipsum}
\usepackage{float}
\usepackage{stfloats}
\usepackage{gensymb}
\usepackage{cite}
\usepackage[font={small}]{caption}

\newcommand\ubar[1]{%
	\underaccent{\bar}{#1}}
\usepackage[hang,flushmargin]{footmisc}
\usepackage{caption}
\usepackage{algpseudocode}

\algtext*{EndWhile}
\algtext*{EndIf}
\algtext*{EndFor}

\DeclareMathOperator*{\argmin}{arg\,min}

\newcommand{\bp}{\begin{proof} \small }
\newcommand{\ep}{\end{proof} \normalsize}
\newcommand{\epx}{\end{proof} \small}
\newcommand{\bpa}{\begin{proofappx} \footnotesize }
\newcommand{\epa}{\end{proofappx} \small }

\newtheorem{proposition}{Proposition}

\newtheorem{lemma}{Lemma}
\newtheorem{assumption}{Assumption}
\newtheorem{definition}{Definition}

\newtheorem*{theorem*}{Theorem}
\newtheorem*{proposition*}{Proposition}
\newtheorem*{corollary*}{Corollary}
\newtheorem*{lemma*}{Lemma}
\newtheorem*{assumption*}{Assumption}
\newtheorem*{definition*}{Definition}
\newtheorem*{claim*}{Claim}

\newcommand{\bm}[1]{\mbox{\boldmath $#1$}}

\newcommand{\be}{\begin{equation}}
\newcommand{\ee}{\end{equation}}
\newcommand{\bs}{\begin{subequations}}
\newcommand{\es}{\end{subequations}}
\newcommand{\bq}{\begin{eqnarray}}
\newcommand{\eq}{\end{eqnarray}}
\newcommand{\bqn}{\begin{eqnarray*}}
\newcommand{\eqn}{\end{eqnarray*}}

\newcommand{\ba}{\left[ \begin{array}}
\newcommand{\ea}{\\ \end{array} \right]}
\newcommand{\ben}{\begin{enumerate}}
\newcommand{\een}{\end{enumerate}}

\def\p{{\boldsymbol{p}}}

\def\real{{\mathchoice%
{\hbox{\rm\setbox1=\hbox{I}\copy1\kern-.45\wd1 R}}
{\hbox{\rm\setbox1=\hbox{I}\copy1\kern-.45\wd1 R}}
{\hbox{\scriptsize\rm\setbox1=\hbox{I}\copy1\kern-.45\wd1 R}}
{\hbox{\scriptsize\rm\setbox1=\hbox{I}\copy1\kern-.45\wd1 R}}}}

\def\Zint{{\mathchoice{\setbox1=\hbox{\sf Z}\copy1\kern-.75\wd1\box1}
{\setbox1=\hbox{\sf Z}\copy1\kern-.75\wd1\box1}
{\setbox1=\hbox{\scriptsize\sf Z}\copy1\kern-.75\wd1\box1}
{\setbox1=\hbox{\scriptsize\sf Z}\copy1\kern-.75\wd1\box1}}}
\newcommand{\complex}{ \hbox{\rm C\kern-0.45em\rule[.07em]{.02em}{.58em}%
\kern 0.43em}}

\makeatletter
\newcommand{\algmargin}{\the\ALG@thistlm}
\makeatother
\newlength{\whilewidth}
\algdef{SE}[parWHILE]{parWhile}{EndparWhile}[1]
{\parbox[t]{\dimexpr\linewidth-\algmargin}{%
		\hangindent\whilewidth\strut\algorithmicwhile\ #1\ \algorithmicdo\strut}}{\algorithmicend\ \algorithmicwhile}%
\algnewcommand{\parState}[1]{\State%
	\parbox[t]{\dimexpr\linewidth-\algmargin}{\strut #1\strut}}

\DeclareMathOperator*{\argmax}{arg\,max} 
\begin{document}
%
\title{Task Replication for Vehicular Cloud: Contextual Combinatorial Bandit with Delayed Feedback}
%
%
%
\author{Lixing Chen,~\IEEEmembership{Student Member,~IEEE},
        ~Jie Xu,~\IEEEmembership{Member,~IEEE}
\thanks{L. Chen and J. Xu are with the Department of Electrical and
	Computer Engineering, University of Miami. Email: \{lx.chen, jiexu\}@miami.edu.}}

\maketitle
\bstctlcite{IEEEexample:BSTcontrol}

\begin{abstract}
Vehicular Cloud Computing (VCC) is a new technological shift which exploits the computation and storage resources on vehicles for computational service provisioning. Spare on-board resources are pooled by a VCC operator, e.g. a roadside unit, to serve computational tasks using the vehicle-as-a-resource framework. This paper investigates timely service provisioning for deadline-constrained tasks in VCC systems by leveraging the task replication technique (i.e., allowing one task to be executed by vehicles). A learning-based algorithm, called DATE-V (Deadline-Aware Task rEplication for Vehicular Cloud), is proposed to address the special issues in VCC systems including uncertainty of vehicle movements, volatile vehicle members, and large vehicle population. The proposed algorithm is developed based on a novel contextual-combinatorial multi-armed bandit learning framework. DATE-V is ``contextual'' because it utilizes side information (context) of vehicles and tasks to infer the completion probability of a task replication under random vehicle movements. DATE-V is ``combinatorial'' because it replicates the received task and sends task replications to multiple vehicles to guarantee the service timeliness. When learning with multi-armed bandit, DATE-V also addresses the practical concern of delayed feedbacks caused by the task transmission/computational delay in using VCC. We rigorously prove that our learning algorithm achieves a sublinear regret bound compared to an oracle algorithm that knows the exact completion probability of any task replications. Simulations are carried out based on real-world vehicle movement traces and the results show that DATE-V significantly outperforms benchmark solutions.
\end{abstract}


%
\IEEEpeerreviewmaketitle

\section{Introduction}
Recent developments in vehicular applications such as autonomous driving, location-specific services and various forms of mobile infotainments are pushing car manufacturers to offer increasingly more advanced and sufficient on-board computing resources. In spite of the phenomenal growth of computing capacities on vehicles, it has been recently noticed that, most of the time, a huge array of on-board computing capacities are chronically underutilized \cite{yan2013security}. A series of recent papers \cite{yan2013security, arif2012datacenter, eltoweissy2010towards} have put forth the vision of Vehicular Cloud Computing (VCC), which pools underutilized vehicular resources (including computing resource, net connection and storage facilities) and rents them to vehicles on the road or other customers, similar to the way in which the resource of conventional cloud is provided but a nontrivial extension of conventional cloud. VCC is actually a paradigm shift from Vehicular Ad Hoc Network (VANET) which includes Vehicle-to-Vehicle (V2V), Vehicle-to-Infrastructure (V2I), and Vehicle-to-everything (V2X) communications. Though originally designed for emergency alerts and collision avoidance, VANET is now merging with intelligent transportation systems, which leads to the advent of intelligent vehicular networks that build ubiquitous vehicular communication environment using different protocols, e.g. dedicated short range communication (DSRC) \cite{kenney2011dedicated}, long-term-evolution-vehicle (LTE-V) \cite{chen2016lte} and 5G technologies \cite{wang2014cellular}. Armed with these components, vehicular networks are evolving into a connected group of smart vehicles.

\begin{figure}[tb]
	\centering
	\includegraphics[width = 0.65\linewidth]{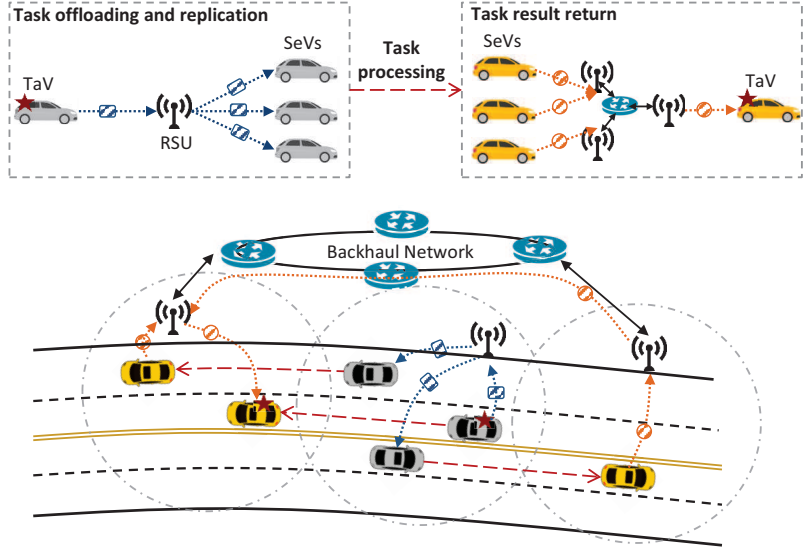}
	\caption{Illustration of task offloading and replication in VCC. \textit{The vehicles in gray denote the location of vehicles when task offloading/replication decisions are made; the vehicles in yellow denote the location of vehicles when the task results are sent back. The red dash arrows denote the traces of vehicles.}}
	\label{fig:illu}
	\vspace{-0.2 in}
\end{figure}
With the pooled vehicular resources, VCC operator uses Vehicle-as-a-Resource (VaaR) to provide computing services to end users (e.g., on-board passengers or pedestrians). Without loss of generality, we define the end user as on-board equipment on vehicles. In this case, the vehicles in a VCC system can be grouped into two categories: Server Vehicle (SeV) and Task Vehicle (TaV). SeVs have surplus computing resources and therefore is pooled by VCC to provide computing service. (On-board equipment of) TaVs have task requests that need to be offloaded for processing. The supply and demand of computing resource are matched by VCC operator, e.g., a roadside unit (RSU), which collects task requests from TaVs and assign them to SeVs. After the task computation, RSUs also collect the results and return them to TaVs. Fig. \ref{fig:illu} shows an illustration of the considered scenario. 

While VCC has offered a basic framework for computing service provisioning, task scheduling policies still need to be carefully designed to guarantee the timeliness of task processing due to the increasing demand on real-time response of vehicular applications. To capture this important feature, we assume the tasks in VCC are deadline-constrained (i.e., the task result must be returned to TaV before a hard deadline, otherwise it becomes useless). Therefore, how to ensure that the tasks can be completed before the deadline becomes the main concern. In this paper, we use the task replication technique to enhance the performance of VCC for deadline-constrained tasks. The key idea of task replication is allowing one task to be offloaded to multiple SeVs and the task is considered completed as long as one of the SeVs processes the task and feeds back the result before the deadline. In this way, the large number of vehicles can be exploited efficiently to provide satisfactory Quality of Service (QoS). 

However, optimally deciding task replications for VCC faces special challenges. First of all, the service delay of a task exhibits large uncertainty due to the unpredictable vehicle trace. For example, as shown in Fig. \ref{fig:illu}, the delay for task result return depends on the locations of TaV and SeVs which decide whether the inter-RSU data transmission is necessary. However, this location information is unknown/uncertain to RSUs when the task replication decisions are made. How to deal with the uncertainty in vehicle mobility will be the most critical issue for task replication in VCC systems. Second, VCC systems are extremely volatile where the vehicles connect and disconnect at any time and the role (TaV or SeV) of vehicles changes frequently. This is very different from existing task scheduling strategies for conventional cloud computing where available servers are fixed in advance. The task replication policy for VCC must be carefully designed to work efficiently with ever-changing system status. Third, the budget constraint is another critical issue that needs to be considered, otherwise, the VCC operator could simply assign a task to all available SeVs. Notice that whether a task is completed by a set of replications follows the \emph{At-Least-One} rule which demonstrates a feature of diminishing rewards. A good task replication policy should stop replicating smartly to ensure that replications are always beneficial for TaVs. 

In this paper, a novel learning algorithm, called DATE-V (Deadline-Aware Task rEplication for Vehicular cloud), is proposed for the replication of deadline-constrained tasks based on Multi-armed Bandit (MAB) framework. We design a novel MAB algorithm, contextual-combinatorial MAB (CC-MAB), to address special challenges in VCC systems. CC-MAB collects context (side information) of computational tasks, TaVs, and SeVs, learns over time the completion probability of task replications with the collected contexts, and exploits learned knowledge to select multiple SeVs for a task request. One salient feature of CC-MAB is that it is able to work with infinitely many vehicles and allow them to appear and disappear at any time. However, the sequential decision making in CC-MAB can be easily interrupted by the stochasticity of task arrival since the new tasks may arrive at RUS before the results of previous tasks are returned (formally termed as \emph{delayed feedback} in MAB problems). To better fit the practical application in VCC, CC-MAB is also extended to learn with the delayed feedback. The key contributions of this paper are summarized as follows: 

1) We first construct a RUS-assisted VCC system and formulate the deadline-constrained task replication problem as a submodular function maximization problem with cardinality constraint. A greedy algorithm is designed to give an oracle solution by assuming that the completion probabilities (i.e. the probability that the task result is returned to TaV before the deadline) of all possible task replications are known a priori.

2) The formulated task replication problem is next considered as a MAB problem. A learning algorithm, DATE-V, is developed within a novel MAB framework, contextual-combinatorial MAB (CC-MAB), which satisfies the special needs of VCC systems. The main advantage of CC-MAB is that it is able to learn efficiently with an infinitely large number of vehicles and ever-changing VCC systems. We analytically bound the loss due to learning, termed \emph{regret}, of DATE-V compared to the oracle benchmark that knows precisely the completion probability of task replications a priori. A regret bound is first provided with non-delayed feedback by assuming the rewards of task replications can be immediately observed, and then extended to the delayed feedback case. The regret upper bounds in both cases are sublinear, which imply that the proposed learning framework produces asymptotically optimal task replication decisions. 
	
3) We carry out extensive simulations using the real-world mobility trace of San Francisco Yellow Cabs. The results show that the proposed DATE-V significantly outperforms other benchmark algorithms.        

The rest of this paper is organized as follows. Section \ref{sec:related_work} reviews related work. Section \ref{sec:system_model} presents the system model and formulates the task replication problem. Section \ref{sec:date-v} designs the learning algorithm DATE-V and gives its performance guarantee. Section \ref{sec:simulation} evaluates the proposed algorithm via simulations, followed by the conclusion in Section \ref{sec:conclusion}.

\section{Related Work}\label{sec:related_work}
Recent efforts have been made to investigate the VCC system. Authors in \cite{eckhoff2011cooperative,liu2011pva} proposed to exploit the spare computing resource on parked cars for task offloading and cooperative sensing. Since the locations of parked cars do not change over time, the task offloading policies for parked cars are similar to those for static cloud servers which have been well investigated \cite{kumar2010cloud,chen2017energy}. There exist works considering moving vehicles in VCC systems \cite{zheng2015smdp,jiang2018task}, where the task scheduling process is assumed to be a Markov Decision Process (MDP). However, the MDP-based approaches usually suffer from the curse of dimensionality and hence cannot be applied to the situation when the number of vehicles is large. By contrast, we solve the task replication problem within a MAB framework, which is a general learning framework and does not rely on additional assumptions on traffic model or scheduling process. The most related work is probably \cite{sun2018learning} where the authors use the MAB framework to help make task offloading decision. While \cite{sun2018learning} only considers a task offloading problem, we consider both task offloading and task replication for VCC systems. More importantly, the MAB algorithm proposed in \cite{sun2018learning} only works with a finite arm set. By contrast, the proposed CC-MAB framework is able to learn with an infinitely large arm set which fits the VCC systems. 

A large body of work has focused on the task replication policy in data retrieval \cite{shah2016redundant} and multi-server data processing systems \cite{wang2015using}. These task replication techniques are usually leveraged to deal with the straggler problem where the service process has a heavy-tail distribution. For example, in \cite{shah2016redundant}, the optimal replication degree, i.e., the number of replicas, is investigated. However, it is not applicable to the VCC system since the servers (vehicles) are not always available and may change across the time.

MAB has been widely studied to address the critical tradeoff between exploration and exploitation in sequential decision making under uncertainty \cite{lai1985asymptotically}. The basic MAB concerns with learning the single optimal arm among a set of candidate arms of a priori unknown rewards by sequentially trying one arm each time and observing its realized noisy reward \cite{auer2002finite}. Combinatorial bandits extends the basic MAB by allowing multiple-play each time \cite{gai2012combinatorial} and contextual bandits extends the basic MAB by considering the context-dependent reward functions \cite{li2010contextual}. While both combinatorial bandits and contextual bandits problems are already much more difficult than the basic MAB problem, this paper tackles the even more difficult contextual-combinatorial MAB problem. Recently, a few other works \cite{li2016contextual,muller2017context} also started to study CC-MAB problems. However, these works make strong assumptions that are not suitable for VCC systems. For instance, \cite{qin2014contextual,li2016contextual} assume that the reward of an individual action is a linear function of the contexts and \cite{muller2017context} assumes a fixed arm set. In our problem, the reward of a replication is unlikely to be a linear function of contexts and, more importantly, the arms may appear and disappear across the time. Delayed feedback \cite{joulani2013online} is another important branch in the MAB family. It concerns with a practical issue that the rewards of arms are not immediately available after the arms are pulled. This issue is also encountered when applying MAB in VCC systems since the transmission/computation delays are incurred to complete the tasks. However, most existing works \cite{joulani2013online, desautels2014parallelizing} on MAB with delayed feedback assume a fixed arm set and hence cannot be applied in our problem.

\section{System Model}\label{sec:system_model}
\subsection{Vehicle Cloud and System Overview}
We consider a Vehicle Cloud Computing (VCC) system where a set of Road Side Units (RSUs) are deployed along the main streets based on certain deployment rules, e.g. improving the overall network performance or maximizing deployment distance \cite{wang2016delivery}. The main functionality of a RSU is receiving tasks from TaVs and dispatching the tasks to appropriate SeVs such that the task results can be returned to TaVs before the deadline. Consider an arbitrary RSU, let $\{1,2,\dots,T\}$ be the sequence of TaVs' tasks received by the RSU. The procedures for completing a task are as follows:

1) \textbf{TaV-to-RSU (T2R) task offloading}: when a TaV issues a task request, it connects to a nearby RSU and offload its task via the wireless connection. The data transmission between TaV and RSU can be easily achieved by existing V2I communication techniques, e.g. DSRC, LTE, and 5G. 

2) \textbf{RSU-to-SeV (R2S) task assignment}: RSU identifies the available SeVs $\mathcal{V}^t$ based on the RSU-to-SeV uplink SINR condition. To ensure the successful task transmission, the SINR of SeV $v$ should be greater than a threshold $\beta$:
\begin{align}\label{eq:SINR}
	\text{SINR}^{(\text{RS})}_v = \frac{P^{\text{R}}(l^{(\text{RS})}_v)^{-\alpha}}{\sigma^2 + I^{\text{RS}}} \geq \beta
\end{align}	
where $P^{\text{R}}$ is the transmission power of RSU, $l^{(\text{RS})}_n$ denotes the distance between RSU and SeV $v$, $\alpha$ is the signal power decay, $\sigma^2$ is the background noise on the frequency channel, $I^{(\text{RS})}$ is the interference, and $\beta$ is the threshold which depends on the wireless network design ($\beta = 0.15$ is recommended for vehicle communication \cite{andrews2009maximizing}).
	
3) \textbf{Task processing}: Once a SeV receives a task, it processes the task with the computing resource on the vehicle. To simplify the system model, it is assumed that the SeV has immediate computing resources to allocate for task processing. Queuing of task at SeVs are not considered in this paper. 
	
4) \textbf{Return results}: After the task is processed, the SeV needs to send the task result back to the TaV. At this time, we denote the RSU associated with SeV as S-RSU and the RSU associated with TaV as T-RSU. SeV first transmits the task result to S-RSU via wireless connection and then S-RSU transmits the results to T-RSU through the backhaul network. When T-RSU receives the result, it sends the result to the TaV. Note that if S-RSU and T-RSU turn out to be the same RSU, the transmission between S-RSU and T-RSU is not performed.

 In our paper, we consider that each task $t$ has a hard deadline requirement $L^t$. A task is completed if the TaV receives the task result before the deadline, otherwise it is failed. The probability of task completion is subject to many uncertain factors, e.g. wireless channel condition and the trace of moving vehicles. Note that whether a SeV can return the task result before the deadline is unknown to the RSU. To increase the completion probability of a task, we allow RSU to send a TaV's task to multiple SeVs. We call each TaV-SeV pair a \emph{replication} of the task $t$ and write all possible replications also as $\mathcal{V}^t$ with slight abuse of notation. 

\subsection{Service Delay and Replication Quality}
Each task $t$ is denoted by a tuple $(x^t, y^t, w^t, b^t, L^t)$ where $x^t$ (in bits) denotes the size of task input data, $y^t$ (in bits) is the size of task result, $w^t$ is the numbers of CPU cycles required to complete the task, $b^t$ is the budget (maximum number of replications) for the task, and $L^t$ is the deadline. Service delay is incurred to complete the task. Let $d^t_v$ denote the service delay of replication $v$. It consists of the following parts:

\subsubsection{T2R task transmission delay} the transmission rate for offloading task $t$ from TaV to RSU can be written as $r^{(\text{TR}),t} = W\log_2(1+\text{SINR}^{(\text{TR})})$. Therefore, the T2R transmission delay can be written as $d^{(\text{TR}),t} = x^t/r^{(\text{TR}),t}$. Note that $d^{(\text{TR}),t}$  is actually revealed to RSU by observing the timestamps of data packets defined by Network Time Protocol (NTP).  
	
\subsubsection{R2S task assignment delay} For simplicity of system model, we assume transmission between RSU and SeV is operated with a fixed transmission rate $r^{(\text{RS}),t}$ by leveraging power/spectrum allocation strategies \cite{liang2017vehicular}. However, our algorithm is compatible with other R2S transmission models which do not give transmission rate exactly. This is because the service delay is modeled as a gray box to the learning algorithm (will be discussed later in this section). 

\subsubsection{Computation delay} let $f^t_v$ be the available CPU frequency allocated by SeV $v$ for task $t$. Then, the computation delay can be simply obtained by $d^{(\text{C}),t}_v = w^t/f^t_v $. Here, we assume that each SeV reports $f^t_v$ to RSU in advance. Again, our algorithm is also able to work when $f^t_v$ is unknown a priori (will be discussed later).

\subsubsection{Result return delay} let $d^{(\text{ST}),t}_v$ be the result return delay of replication $v$. It consists of the following parts: 1) the transmission delay between SeV $v$ and R-SRU $d^{(\text{SR}),t}_v = y^t/r^{(\text{SR}),t}_v$, where $r^{(\text{SR}),t}_v = W\log_2(1+\text{SINR}^{(\text{SR})}_v)$ is the transmission rate. 2) the backhaul transmission delay between S-RSU and T-RSU $d^{(\text{RR}),t}_v = y^tg^t_v +h^t_v$, where $g^t_v$ is the backhaul transmission rate and $h^t_v$ be the round trip time when sending back the result of replication $v$. If S-RUS and T-RSU are the same, then $d^{(\text{RR}),t}_v = 0$; 3) the delay for transmitting results between RSU and TaV $d^{({\text{RT}}),t}_v = y^t/r^{(\text{RT}),t}_v$ where $r^{(\text{RT}),t}_v$ is the fixed transmission rate operated by RSU for RSU-to-Vehicle data transmission. The result return delay of replication $v$ can be obtained as $d^{(\text{ST}),t}_v = d^{(\text{SR}),t}_v + d^{(\text{RR}),t}_v + d^{({\text{RT}}),t}_v$.
	
Therefore, the total service delay of replication $v$ is $d^t_v = d^{(\text{TR}),t} + d^{(\text{RS}),t} + d^{(\text{C}),t}_v + d^{(\text{ST}),t}_v$. One can immediately see that the delay model for a replication is a ``gray box'' to the RSU operators: while some parts of service delay are revealed to the RSU (e.g. $d^{(\text{TR}),t}$, $d^{(\text{RS}),t}$, and $d^{(\text{C}),t}_v$), the result return delay $d^{({\text{RT}}),t}$ is unknown to the RSU due to the uncertainty in vehicle movement and backhaul network condition. If TaV receives the result of replication $v$ before the deadline, i.e. $d^t_v \leq L^t$, then replication $v$ is considered as successfully executed. We define \emph{quality} of replication $v$ as $q^t_v = \textbf{1}\{d^t_v \leq L^t\}$, where $\textbf{1}\{\cdot\}$ is the indicator function. Let $\mu^t_v = \mathbb{E}\{q^t_v\} = \Pr\{d^t_v \leq L^t\}$ be the expected quality of replication $v$. Since $d^{(\text{TR}),t}$, $d^{(\text{RS}),t}$, and $d^{(\text{C}),t}_v$ are known to the RSU, the expected quality of replication $v$ can be written as $\mu^t_v = \Pr\{d^{(\text{ST}),t}_v \leq L^t - d^{(\text{TR}),t} - d^{(\text{RS}),t}_v - d^{(\text{C}),t}_v\}$. 
 
\textbf{Remark}: The service delay ``gray box'' can be changed to other configurations, e.g., RSUs do not use fixed transmission rate or the SeVs do not report the CPU frequency allocated for a task. In this case, $d^{(\text{RS}),t}$ and $d^{(\text{C}),t}_v$ become unknown, and the expected quality of replication $v$ can be written as $\mu^t_v = \Pr\{d^{(\text{ST}),t}_v + d^{(\text{RS}),t}_v + d^{(\text{C})}_v \leq L^t - d^{(\text{TR}),t}\}$. Our method is able to work even if the service delay is a ``black box''.     
 
\subsection{Problem Formulation}
For each task $t$, the RSU picks a subset of replications from all available replications $\mathcal{V}^t$ for task $t$, and we call the subset $\mathcal{A}^t \subseteq \mathcal{V}^t$ the replication decision for task $t$. The reward $r(\mathcal{A}^t)$ achieved by the selected replications in $\mathcal{A}^t$ is defined as:
\begin{equation}\label{eq:reward_func_A}
	r(\mathcal{A}^t) = \left\{
	\begin{split}
	& 1 - \eta \cdot  |\mathcal{A}^t|, ~~\text{if}~~\exists v \in\mathcal{A}^t, q^t_v = 1\\
	& - \eta \cdot  |\mathcal{A}^t|, ~~\text{if}~~\forall v \in\mathcal{A}^t, q^t_v = 0
	\end{split}\right.
\end{equation}
The term $\eta \cdot |\mathcal{A}^t|$ in \eqref{eq:reward_func_A} captures the cost of replication decision $\mathcal{A}^t$, where $\eta$ is the unit cost for one replication. By applying the \emph{At-Least-One} probabilistic rule, we can write the expected reward of replication decision $u(\bm{\mu}^t, \mathcal{A}^t) = \mathbb{E}[r(\mathcal{A}^t)]$ as: 
\begin{align} \label{eq:reward_func}
u(\bm{\mu}^t,\mathcal{A}^t) = \left(1 - \prod\nolimits_{v\in \mathcal{A}^t}(1-\mu^t_v)\right) - \eta \cdot |\mathcal{A}^t| 
\end{align}
where the first term in \eqref{eq:reward_func} denotes the probability that task $t$ is completed by at least one replication in $\mathcal{A}^t$. Consider an arbitrary sequence of tasks $\{1,2,\dots,T\}$ that arrive at a RSU. The RSU makes task replication decisions $\mathcal{A}^t$ for each task $t$ which aims to maximize the expected cumulative rewards:
 \begin{subequations}
	\begin{align}
	\textbf{P1:}~~~& \max\nolimits_{\{\mathcal{A}^t\}^T_{t=1}} ~~ \sum\nolimits_{t=1}^T u(\bm{\mu}^t,\mathcal{A}^t) \\
	\text{s.t.} ~~&\mathcal{A}^t\subseteq \mathcal{V}^t, \forall t = 1, 2, \dots, T \\
	&|\mathcal{A}^t| \leq b^t, \forall t = 1, 2, \dots, T \label{const:budget}
	\end{align}
\end{subequations}
where the constraint \eqref{const:budget} indicates that the number of replication in $\mathcal{A}^t$ should not exceed the budget $b^t$ of task $t$. The problem in \textbf{P1} can be decouple into $T$ independent subproblems, one for each task $t$ as follows: 
\begin{align}\label{op:perslot}
\textbf{P2:}~~~\max_{\mathcal{A}^t\subseteq \mathcal{V}^t,|\mathcal{A}^t| \leq b^t} ~~ 1 - \prod_{v\in \mathcal{A}^t}(1-\mu^t_v) - \eta \cdot |\mathcal{A}^t|
\end{align}
The objective in \textbf{P2} exhibits a property of \emph{submodularity}: the total reward achieved by the selected replication is not a simple sum of individual qualities but demonstrates a feature of diminishing returns determined by the \emph{At-Least-One} rule. The formal definition of submodularity is given below.
\begin{definition}[Submodularity]
	Let $\mathcal{V}$ be the universe replication set. For all possible subsets $\mathcal{A} \subseteq \mathcal{B} \subseteq \mathcal{V}$ and any replication $v \notin\mathcal{B}$, if a reward function $u(\cdot,\cdot)$ satisfies $u(\bm{\mu}^t,\{v\}\cup\mathcal{A}) - u(\bm{\mu}^t,\mathcal{A}) \geq u(\bm{\mu}^t,\{v\}\cup\mathcal{B}) - u(\bm{\mu}^t,\mathcal{B}). $
	Then, $u(\cdot,\cdot)$ is submodular.
\end{definition}
We for now assume that there was an omniscient oracle that knows the expected quality of each possible replication $\mu^t_v, v \in \mathcal{V}^t$. Then \textbf{P2} becomes a submodular function maximization problem with cardinality constraint, which is a well-studied topic and can be efficiently solved by the greedy algorithm presented in Algorithm \ref{alg:greedy}. To facilitate the solution presentation, we define the marginal reward of replication $v$:
\begin{definition}[Marginal Reward]
	Consider a task $t$, let $\mathcal{A}^t\subseteq \mathcal{V}^t$ be a subset of replications and let $v^\prime$ be an available replication. Define the \emph{marginal reward} of adding replication $v^\prime$ to $\mathcal{A}^t$ as $\Delta(\bm{\mu},\{v^\prime\} | \mathcal{A}^t) = u(\bm{\mu}, \{v^\prime\}\cup\mathcal{A}^t) - u(\bm{\mu},\mathcal{A}^t)$.
\end{definition} 

\begin{algorithm}[htb]
	\caption{Greedy Algorithm}
	\begin{algorithmic}[1] 
		\State \textbf{Input}: $\mathcal{V}^t$, $b^t$, $\mu^t_v, \forall v \in \mathcal{V}^t$.
		\State \textbf{Initialization}: $\mathcal{A}_0 \gets \emptyset, k \gets 0$;
		\While {$k \leq b^t$}:
		\State $k = k + 1$;
		\State select $v_k = \argmax_{v\in\mathcal{V}^t\backslash\mathcal{A}_{k-1}} \Delta(\bm{\mu}^t, \{v\} | \mathcal{A}_{k-1})$;
		\If {$\Delta(\bm{\mu}^t, \{v_k\} | \mathcal{A}_{k-1})>0$:} $\mathcal{A}_{k} = \mathcal{A}_{k-1} \cup \{v_k\}$;
		\Else: stop;
		\EndIf 
		\EndWhile
		\State \textbf{return} $\mathcal{A}_{k}$
	\end{algorithmic}\label{alg:greedy} 
\end{algorithm}

The greedy algorithm works in an iterative manner. In each iteration $k$, a replication $v_k\in\mathcal{V}^t\backslash\mathcal{A}_{k-1}$ is selected such that the marginal reward is maximized given $\mathcal{A}_{k-1} = \cup^{k-1}_{i=1}v_i$. In the general case, the greedy algorithm guarantees no less than $(1-1/e)$ of the optimum in only polynomial runtime. However, for our problem in \textbf{P2}, the greedy algorithm actually gives the optimal solution, which is proved in Proposition \ref{prop:greedy_opt}.
\begin{proposition} [Optimality of Greedy Algorithm]\label{prop:greedy_opt}
	For an arbitrary task $t$, the task replication decision to the $t$-th subproblem derived by the greedy algorithm is optimal.
\end{proposition}	

\begin{proof}
see in online Appendix \ref{proof:prop:greedy_op} \cite{onlineappendix}.
\end{proof}

Let $\mathcal{A}^{*,t}$ be the optimal replication decision to the per-slot problem for task $t$. Therefore, the optimal solution for \textbf{P1} is $\{\mathcal{A}^{*,t}\}_{t = 1}^{T}$. Since this optimal solution is obtained by an oracle, we call it the oracle solution. Let $\{\mathcal{A}^{t}\}^T_{t=1}$ be the replication decisions derived by a certain algorithm. The performance of this algorithm is evaluated by comparing its loss with respect to the oracle algorithm. This loss is called the \emph{regret} of the algorithm which is formally defined as 
\begin{align}
	R(T) = \mathbb{E}\left[ \sum\nolimits_{t=1}^{T}  \left(r(\mathcal{A}^{*,t}) - r(\mathcal{A}^t)\right) \right]
\end{align}
which is equivalent to $R(T) = \sum_{t=1}^{T} u(\bm{\mu}^t,\mathcal{A}^{*,t}) - u(\bm{\mu}^t,\mathcal{A}^t)$. 

In the above, we have discussed the oracle solution to \textbf{P1} by assuming that the expected quality of replications is known to the RSU. However, in the real VCC application, it is difficult, if not impossible, to know in advance the replication qualities precisely due to the uncertainty in vehicle movement and network conditions. In this case, the replication decisions cannot be easily derived by the greedy algorithm alone. In the next section, we put the task replication problem into a contextual-combinatorial MAB (CC-MAB) framework, such that a RSU is able to learn the expected quality of task replications over time by observing the contexts of replications and then make smart replication decisions.   

\section{CC-MAB for Deadline-Aware Task Replication}\label{sec:date-v}
Whether a replication can be completed depends on many factors which are collectively referred to as \emph{context}. For example, relevant factors can be task information (e.g. data size of task input and result affects the delay during transmission), vehicle information (e.g. speeds of TaV and SeVs influence vehicle locations when sending back the task results, and therefore determine whether the inter-RSU transmission is necessary), road conditions (e.g. high vehicle density causes high wireless transmission interference and therefore increases the transmission delay). This categorization is clearly not exhaustive and the impact of each single context on the replication quality is unknown a priori. Our algorithm will learn to discover the underlying connection between such context and replication quality, thereby facilitating the task replication decision making.

\subsection{Context-aware Task Replication}
Let $\Phi_{\text{T}}$ be the context space of tasks which includes the task information (e.g. size of input/result data, deadline, etc.) and TaVs' vehicle information (e.g. speed, location, and available computational resources). Let $\Phi_{\text{S}}$ be the context space of SeVs' vehicle information. The RSU sets the joint space $\Phi = \Phi_{\text{T}} \times \Phi_{\text{S}}$ as the context space of replications. The context space $\Phi$ is assumed to be bounded and hence can be denoted as $\Phi = [0,1]^D$ without loss of generality, where $D$ is the number of dimensions of context space $\Phi$. Since the service delay of replication $v\in\mathcal{V}^t$ is now parameterized by its context $\phi^t_v$, we write a quality of a replication $v$ as $q^t_v(\phi^t_v) = \textbf{1}\{d^t_v(\phi^t_v) \leq L^t\}$ and its expected value as $\mu^t_v(\phi^t_v)$. Let $\bm{\mu}^t = \{\mu_v^t(\phi^t_v)\}_{v\in\mathcal{V}^t}$ collect all the context-specific replication qualities.

Now, we are ready to formulate the task replication problem as a CC-MAB problem. For each task $t$, the RSU operates sequentially as follows: (i) upon receiving the task request, the RSU lists all possible replications $\mathcal{V}^t$ and observes the context $\phi^t_v \in \Phi$ for each replication $v\in\mathcal{V}^t$. Let $\bm{\phi}^t = \{\phi^t_v\}_{v\in\mathcal{V}^t}$ collect all replications' context. (ii) the RSU selects a subset of replications based on the observed context $\bm{\phi}^t$ and the knowledge learned from the previous tasks. (iii) the RSU sends the task replications to selected SeVs and then collects results when the task is processed. (iv) The RSU sends the task result back to TaV and observes the quality of selected replications. The observed qualities will be used to update the current knowledge. Yet, notice that the quality task $t$'s replication may not be immediately observed before the arrival of task $t+1$ due to the transmission/computation delays incurred by VCC, which causes the problem of delayed feedback in CC-MAB. For the ease of presentation and explanation, we assume for now that the qualities of replications for task $t$ are observed before the arrival of next task $t+1$, and therefore the feedback of CC-MAB is non-delayed. The more practical case of delayed-feedback will be discussed later in this section. 

\subsection{DATE-V with Non-delayed Feedback}
DATE-V (Deadline-Aware Task rEplication in Vehicle Cloud) (in Algorithm \ref{alg:ccmab}) is developed  based on the CC-MAB framework. In the initialization phase, DATE-V creates a partition $\mathcal{P}_T$ on the contexts space $\Phi$, which splits $\Phi$ into $(h_T)^D$ sets based on the given time horizon $T$. These sets are given by $D$-dimensional hypercubes of identical size $\frac{1}{h_T} \times \dots \times \frac{1}{h_T}$. Here, $h_T$ is a parameter to be designed to determine the number of hypercubes in $\mathcal{P}_T$. Additionally, the RSU keeps a counter $C^t(p)$ for each hypercube $p \in \mathcal{P}_T$ which records the number of selected replications that have context $\phi^t_v$ falling in hypercube $p$ before receiving task $t$. Fig. \ref{fig:ctxt_space} offers an illustration of context partition and counter update. Moreover, the algorithm also keeps an estimated quality $\hat{\mu}^t(p)$ for each hypercube. Let $\mathcal{Q}^t(p) = \{q(\phi^\tau_v)|~\phi^\tau_v \in p, v \in \mathcal{A}^\tau, \tau = 1,\dots,t-1\}$ be the set of observed qualities of replications with context in $p$. Then, the estimated quality for replications with context $\phi^t_v \in p$ is $\hat{\mu}^t(p) = \frac{1}{C^t(p)}\sum\nolimits_{q \in \mathcal{Q}^t(p)} q$.
\begin{figure}[tb]
	\centering
	\includegraphics[width = 0.55\linewidth]{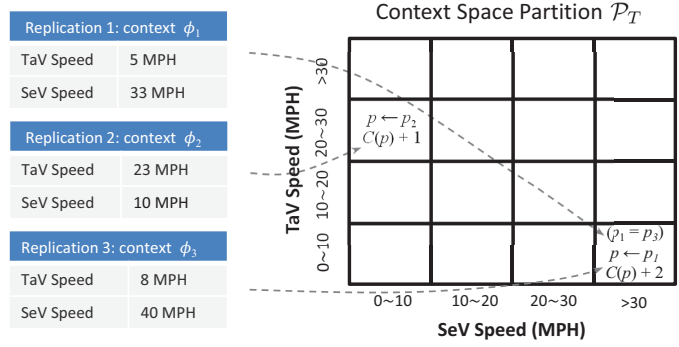}
	\caption{Illustration of context space partition and counter update.}
	\label{fig:ctxt_space}
	\vspace{-0.15 in}
\end{figure}

For each task $t$, DATE-V performs the following steps: the contexts of all possible replications $\bm{\phi}^t = \{\phi^t_v\}_{v\in\mathcal{V}^t}$ are observed. For each context $\phi^t_v$, the algorithm determines a hypercube $p^t_v \in \mathcal{P}_T$ such that $\phi^t_v \in p^t_v$ holds. The collection of these hypercubes for task $t$ is denoted by $\p^t = \{p^t_v\}_{v \in \mathcal{V}^t}$. Then the algorithm checks if there exist hypercubes $p\in \p^t$ that have not been explored sufficiently often. For this purpose, we define the \emph{under-explored} hypercubes for task $t$ as:
\begin{align}
\mathcal{P}^{\text{ue},t}_T  \triangleq \left\{p\in \mathcal{P}_T \mid \exists~v\in\mathcal{V}^t, \phi^t_v \in p,  C^t(p)\leq K(t)\right\}
\end{align}
where $K(t)$ is a deterministic, monotonically increasing control function that needs to be designed by CC-MAB. In addition, we collect the replications that fall in the under-explored hypercubes in $\mathcal{V}^{\text{ue},t} \triangleq \{v\in\mathcal{V}^t \mid p^t_v\in \mathcal{P}^{\text{ue},t}_T\}$. 

Depending on the under-explored replications $\mathcal{V}^{\text{ue},t}$ for task $t$, DATE-V can either be in an exploration phase or an exploitation phase. If $\mathcal{V}^{\text{ue},t}$ is non-empty, DATE-V enters an exploration phase. Let $z = |\mathcal{V}^{\text{ue},t}|$ be the size of under-explored replications. If the set $\mathcal{V}^{\text{ue},t}$ contains at least $b^t$ replications ($z \geq b^t$), then DATE-V randomly selects $b^t$ replications from $\mathcal{V}^{\text{ue},t}$. If $\mathcal{V}^{\text{ue},t}$ contains fewer than $b^t$ replications ($z < b^t$), then DATE-V selects all $z$ replications from $\mathcal{V}^{\text{ue},t}$. Since the budget $b^t$ is not fully utilized, the remaining $(b^t-z)$ replications are picked using Greedy algorithm (Algorithm \ref{alg:greedy}) with estimated qualities $\hat{\bm{\mu}}^t$:
\begin{align}\label{eq:semi_explore}
v_k = \argmax_{v\in\mathcal{V}^t\backslash\{\mathcal{V}^{\text{ue},t}\cup\mathcal{A}_{k-1}\}} \Delta(\hat{\bm{\mu}}^t, \{v\} | \{\mathcal{A}_{k-1}\cup\mathcal{V}^{\text{ue},t}\}),
\end{align}
where $k = 1,\dots, (b^t-z)$ and $\mathcal{A}_{k-1} = \{v_i\}_{i=1}^{k-1}$. If $\mathcal{V}^{\text{ue},t}$ is empty, DATE-V enters an exploitation phase. It selects up to $b^t$ replications using Greedy algorithm with estimated qualities:
\begin{align}\label{eq:exploit}
v_k = \argmax_{v\in\mathcal{V}^t\backslash\mathcal{A}_{k-1}} \Delta(\hat{\bm{\mu}}^t, \{v\} | \mathcal{A}_{k-1}), ~~k = 1,\dots, b^t
\end{align} 
After selecting the replications, DATE-V observes the qualities realized by selected replications and then updates the estimated quality and the counter of each hypercube in $\p^t$. Note that the task index $t$ of replication quality estimation and counter is drop in the pseudo-code (Line \ref{line:update}) since previous values of counters and quality estimations do not need to be stored.
\begin{algorithm}[tb]
	\caption{DATE-V}
	\begin{algorithmic}[1] 
		\State \textbf{Input}: $T$, $h_{T}$, $K(t)$, $\bm{\phi}^t$.
		\State \textbf{Initialization}: $\mathcal{P}_{T}$; $C^0(p) = 0, \hat{\mu}(p) = 0, \forall p \in \mathcal{P}_{T}$;
		\For {$t=1,\dots,T$}:
		\State Observe replications $\mathcal{V}^t$ and contexts $\bm{\phi}^t=\{\phi^t_{v}\}_{v\in\mathcal{V}^t}$;
		\parState{Find $\p^t=\{p^t_{v}\}_{v\in\mathcal{V}^t}, p^t_{v} \in\mathcal{P}_{T}$ such that $\phi^t_{v}\in p^t_{v}$;}
		\parState{Identify  $\mathcal{P}^{\text{ue},t}$ and $\mathcal{V}^{\text{ue},t}$; let $z=|\mathcal{V}^{\text{ue},t}|$;}
		\If {$\mathcal{P}^{\text{ue},t} \neq \emptyset$}: \Comment{\textit{Exploration}}
		\If{$z \geq b^t$}: \parState{$\mathcal{A}^t \leftarrow$ randomly pick $b^t$ replications from $\mathcal{V}^{\text{ue},t}$;}
		\Else : \parState{$\mathcal{A}^t \leftarrow$ pick $z$ replications in $\mathcal{V}^{\text{ue},t}$ and other $(b^t-z)$ as in \eqref{eq:semi_explore};}
		\EndIf
		\Else:  $\mathcal{A}^t \leftarrow$ pick $b^t$ arms as in \eqref{eq:exploit};  \Comment{\textit{Exploitation}}
		\EndIf
		\For {each replication $v\in\mathcal{A}^t$}:
		\State Observe quality $q^t_v$ of replication $v$;
		\State Update estimation: $\hat{\mu}(p^t_{v})=\frac{\hat{\mu}(p^t_{v})C(p^t_{v})+q^t_{v}}{C(p^t_{v})+1}$; \label{line:update}
		\State Update counters: $C(p^t_{v})=C(p^t_{v})+1$; \label{line:counter_update_o}
		\EndFor
		\EndFor
	\end{algorithmic} \label{alg:ccmab}
\end{algorithm}

It remains to design the parameter $h_T$ and the control policy $K(t)$ to achieve a sublinear regret in the time horizon $T$, i.e., $R(T) = O(T^\gamma), \gamma < 1$, such that DATE-V guarantees an asymptotically optimal performance ($\lim_{T\to\infty} R(T)/T=0$). 

\subsection{Parameter Design and Regret Analysis}
In this section, we design the algorithm parameters $h_T$ and $K(t)$ and give a corresponding upper bound for the regret incurred by DATE-V. The regret analysis is carried out based on the natural assumption that the expected qualities of arms are similar if they have similar context \cite{muller2017context}. This assumption is formalized by the H\"{o}lder condition as follows: 
\begin{assumption}[H\"{o}lder Condition] \label{holder}
	There exists $L>0$, $\alpha>0$ such that for any two contexts $\phi,\phi^\prime\in \Phi$, it holds that $|\mu(\phi)-\mu(\phi^\prime)| \leq L \|\phi-\phi^\prime\|^{\alpha}$, where $\|\cdot\|$ is the Euclidean norm.
\end{assumption}
Assumption \ref{holder} is needed for the regret analysis, but it should be noted that DATE-V can also be applied if this assumption does not hold. However, a regret bound might not be guaranteed in this case. Now, we set $h_T = \lceil T^{\frac{1}{3\alpha+D}}\rceil$ for the context space partition, and $K(t) = t^{\frac{2\alpha}{3\alpha+D}}\log(t)$ in each time slot $t$ for identifying the under-explored hypercubes and arms. Then, we have a sublinear regret upper bound of DATE-V:    

\begin{proposition}[Regret Upper Bound] \label{theo:upper_bound}
	Let $K(t) = t^{\frac{2\alpha}{3\alpha+D}}\log(t)$ and $h_T = \lceil T^{\frac{1}{3\alpha+D}}\rceil$. If H\"{o}lder condition holds true, the regret $R(T)$ is bounded by 
	\begin{align*}
	R(T) \leq &(1+\eta B)2^{D}\left(\log(T)T^{\frac{2\alpha+D}{3\alpha+D}} + T^{\frac{D}{3\alpha+D}}\right)  \\
	& +(1+\eta B)\left(\sum\nolimits_{k=1}^B{V^{\max} \choose k}\right) \frac{\pi^2}{3} \\
	& + \left(3BLD^{\alpha/2}+\frac{2B+2BLD^{\alpha/2}}{(2\alpha+D)/(3\alpha+D)} \right) T^{\frac{2\alpha+D}{3\alpha+D}}
	\end{align*}
	where $B = \max(b^1,\dots,b^T)$ is the maximum possible budget for a task. The leading order of the regret $R(T)$ is $O\left( (1+\eta B)2^{D} T^{\frac{2\alpha+D}{3\alpha+D}}\log(T)\right)$, which is sublinear.
\end{proposition}
\begin{proof}
	See in online Appendix \ref{proof:theo:upper_bound} \cite{onlineappendix}. 
\end{proof}  
The regret upper bound given in Proposition \ref{theo:upper_bound} is sublinear for a sequence of tasks $\{1,2,\dots,T\}$. In addition, the regret bound is valid for any finite task number, thereby providing a bound on the performance loss for any finite $T$. Therefore, this Proposition also can be used to characterize the convergence speed of DATE-V.

\subsection{DATE-V with Delayed Feedback}
We have evaluated the performance of DAVE-V with non-delayed feedback. However, the non-delayed feedback assumption can be easily violated in application since the RSU can observe the qualities of replications only after the task results are returned, yet at this time new task requests may have already arrived. Therefore, in the following, we analyze the performance of DATE-V with delayed feedback. 

For an arbitrary task $t$, DAVE-V has counters $C^t(p), p\in \p^t$ that counts the number of selected replications with context in $p$. Since the feedback is delayed, the number of observed qualities may be less than the number of selected replications. Therefore, we introduce a new counter $M^t(p)$ to record the number of observed qualities for the replications with context in hypercube $p$ before receiving task $t$. Clearly, we will have $M^t(p) \leq C^t(p)$. Let $\mathcal{Q}_M^t(p)$ be the set of observed qualities, the estimated quality is now $\hat{\mu}(p) = \frac{1}{M^t(p)}\sum_{q \in \mathcal{Q}_M^t(p)} q$. 

Now, we compare the performances of DAVE-V under non-delayed feedback and delayed feedback cases by analyzing the exploration and exploitation phases separately. We first consider the exploration phase of DATE-V in a delayed feedback setting. Whether DATE-V will enter the exploration for task $t$ is determined by the counters $C^t(p), \forall p \in \p^t$ and does not depend on the number of observed qualities. Therefore, the regrets incurred by the exploration in the non-delayed feedback and delayed feedback cases are the same. Next, we consider the exploitation phase of DATE-V with delayed feedback. For a task $t$, if its counters satisfy $C^t(p) > K(t), \forall p \in \p^t$, then DATE-V enters the exploitation phase. Due to the delayed feedback, we have two cases for the exploitation: i) the number of observed qualities satisfies $K(t)<M^t(p)\leq C^t(p), \forall p \in \p^t$. In this case, though there are qualities remaining unobserved, the number of observed qualities is larger than $K(t)$. Therefore, exploiting the estimated qualities guarantees the regret bound as proved in the non-delayed feedback case. ii) the number of observed qualities in $p$ satisfies $M^t(p)<K(t)\leq C^t(p), \exists p \in \p^t$. Since the number of observed qualities $M^t(p), \exists p\in \mathcal{P}_T$ is less than $K(t)$, using $\mu^t(p)$ for task replication cannot guarantee the regret bound in the exploitation. We call the exploitation phase $t$ with $M^t(p)<K(t)\leq C^t(p), \exists p \in \p^t$ as \emph{mis-exploitation}. To bound the regret of DATE-V with delayed feedback, we only need to consider the extra regret in mis-exploitation.

\begin{proposition}\label{theo:upper_bound_DF}
	If DATE-V is run with the parameters given in Proposition \ref{theo:upper_bound}, the regret due to mis-exploitation is
	\begin{align*}
		R_m(T)\leq  \lambda L^{\max} (1+\eta B)(T^{\frac{2\alpha}{3\alpha + D}}\log(T) + 1).
	\end{align*}
	where $\lambda$ is the task arrival rate and $L^{\max}$ is the maximum task deadline. The regret of DATE-V with delayed-feedback $R^\prime (T)$ is bounded by $R^\prime(T) \leq  R(T) + R_m(T)$, where $R(T)$  is the regret upper bound of DATE-V with non-delayed feedback. 
\end{proposition}
\begin{proof}
	See in online Appendix \ref{proof:theo:upper_bound_DF} \cite{onlineappendix}.
\end{proof}

Proposition \ref{theo:upper_bound_DF} shows that the regret of DATE-V with delayed feedback is the regret of DATE-V with non-delayed feedback plus an additional term which grows with the increase task arrival. Note that this additional term is still sublinear in $T$ which mean the regret of DATE-V with delayed feedback is still sublinear. Moreover, we see that the leading order of $R^\prime(T)$ is the same as that of $R(T)$.

\section{Simulation}\label{sec:simulation}
\subsection{Simulation Setup}
Our simulation uses the mobility trace of San Francisco Yellow Cabs \cite{epfl-mobility-20090224}. It records the GPS coordinates of 550 cabs, logged approximately every 45 seconds, over a period of 30 days, in the San Francisco Bay Area. These cab traces are used to simulate the vehicle movement in the VCC system. We focus on an area of coordinate from $37^{\degree}74^\prime$N to $37^{\degree}76^\prime$N, $122^{\degree}39^\prime$W to $122^{\degree}24^\prime$W. Fig. \ref{fig:RSUdep} depicts a portion of all cab traces in this area, which at the same time shows the road layout. We deploy a total of 12 RSU along the main roads. The distance between two neighbor RSUs is set around 200m. The maximum coverage of a RSU is set as 300m such that most vehicles in this area can access at least one RSU. A RSU is randomly selected as an example to run the proposed algorithm. 
\begin{figure}[b]
	\vspace{-0.1 in}
	\centering
	\includegraphics[width=0.5\linewidth]{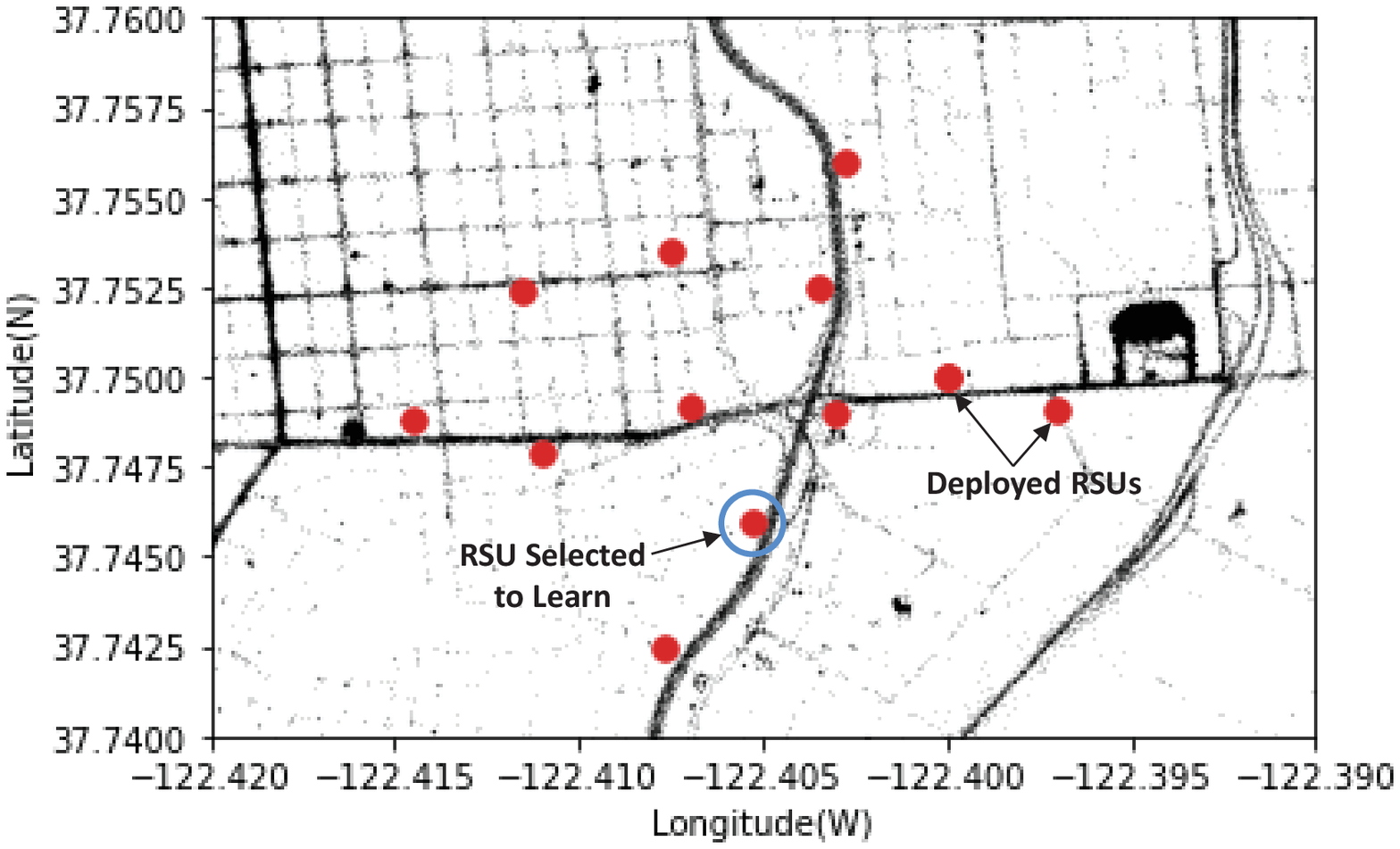}
	\caption{Road layout and RSU deployment.}
	\label{fig:RSUdep}
\end{figure}
For simplicity, we assume the tasks from TaVs are of the same type with the input data size $x^t = 1$Mb, the task result size $y^t = 0.5$Mb and the required CPU cycles $w^t = 200$M. The deadline $L^t$ for each task $t$ is randomly chosen from $[1,2.5]$sec. The RSU-to-Vehicle data transmission operates on fixed transmission rate 3Mbps. The Vehicle-to-RSU transmission rate is determined by the Shannon Capacity where the bandwidth $W = 10$MHz, transmission power of vehicles is 10dBm, noise power $\sigma^2 = -172$dBm. The backhaul transmission rate is chosen from $g^t \in [0.5,1.5]$Mbps. The round trip time for backhaul transmission is chosen from $[20, 300]$ms. We collect speeds and locations of TaVs/SeVs, and the task deadlines as \emph{context}. Fig. \ref{fig:context_reward} depicts the impact of \emph{location} and \emph{task deadline} on the replication quality, which shows that the quality of a replication is very related to its context.
\begin{figure}[tb]
	\centering
	\vspace{-0.1 in}
	\includegraphics[width=0.5\linewidth]{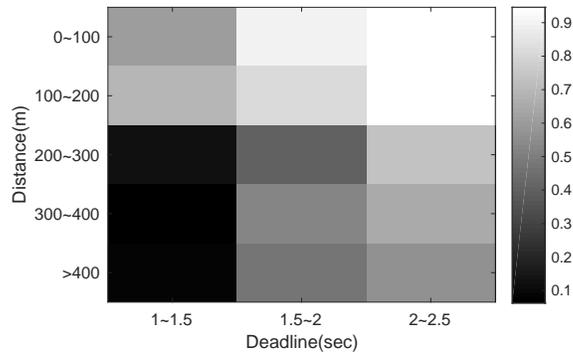}
	\caption{Impact of the context on replication quality. \textit{We focus on the vehicle locations and task deadlines; the location of TaV and SeV are converted to distance. Replications with smaller distance tend to have higher quality.}}
	\vspace{-0.15 in}
	\label{fig:context_reward}
\end{figure}
DATE-V is compared with the following benchmarks:\\
\textbf{1) Oracle}: Oracle knows precisely the expected quality of each replication before making task replication decisions. For each task, Oracle selects replications based on expected qualities using the greedy algorithm presented in Algorithm \ref{alg:greedy}.\\
\textbf{2)} \textbf{mLinUCB}: LinUCB \cite{li2010contextual} is a contextual bandit algorithm which recommends exactly one arm in each round. To select multiple replications, mLinUCB repeats LinUCB algorithm $b^t$ times in each round. By sequentially removing selected replications, we ensure that the $b^t$ replications returned by mLinUCB are distinct in for each task $t$.\\
\textbf{3) UCB}: UCB algorithm \cite{auer2002finite} is a classical MAB algorithm (non-contextual and non-combinatorial) that achieves the logarithmic regret bound. Similar to mLinUCB, we repeat UCB $b^t$ times to select multiple replications for each task.\\
\textbf{4) Random}: The Random algorithm picks $b^t$ replications randomly from the available replications for each task $t$.

	\begin{figure}[t]
		\centering
		\includegraphics[width=0.5\linewidth]{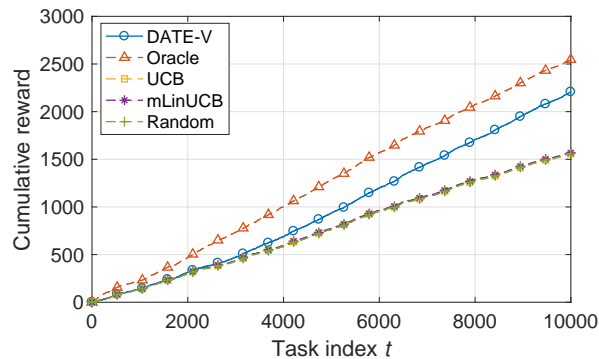}
		\caption{Comparison on cumulative reward}
		\label{fig:cum_reward}
	\end{figure}%
	\begin{figure}[t]
		\centering
		\includegraphics[width=0.5\linewidth]{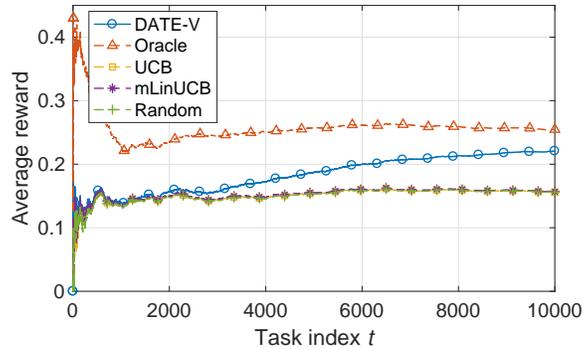}
		\caption{Comparison on average reward}
		\label{fig:ave_reward}
	\end{figure}%
	\begin{figure}[t]
		\centering
		\includegraphics[width=0.5\linewidth]{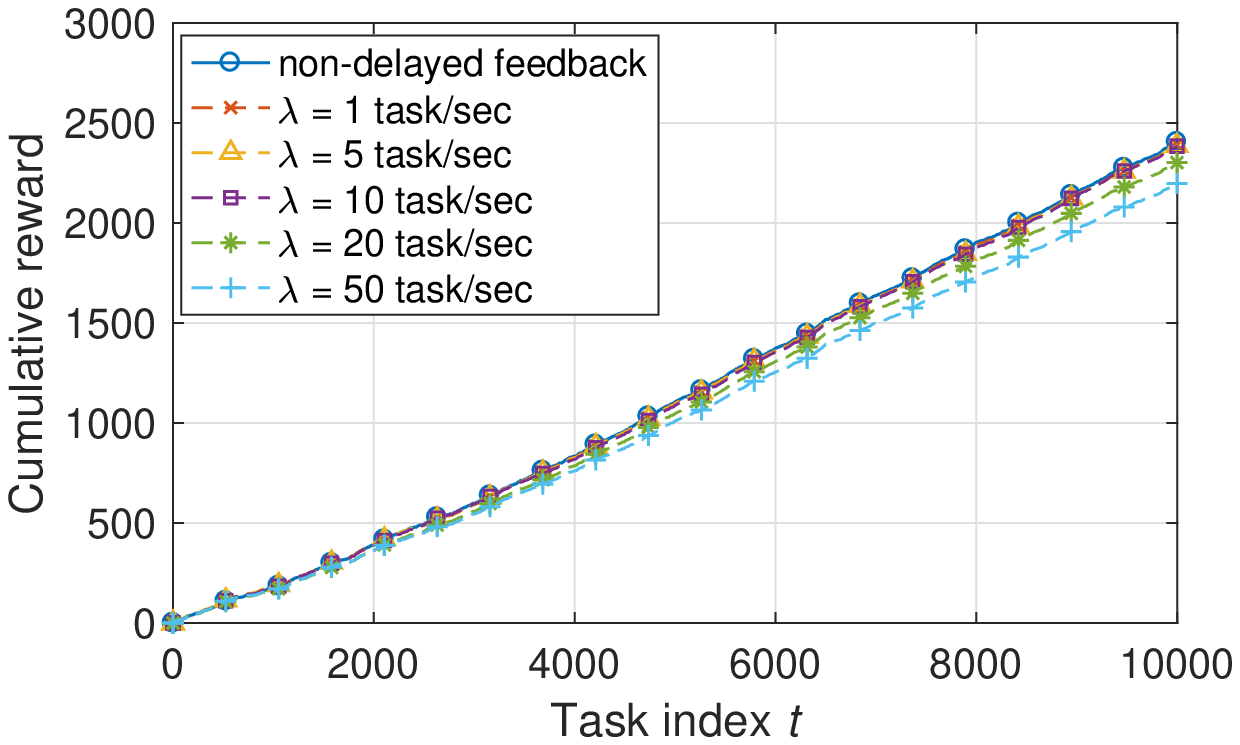}
		\caption{Impact of task arrival rate}
		\label{fig:vary_lambda}
	\end{figure}%

\subsection{Performance Comparison}
Fig. \ref{fig:cum_reward} shows the cumulative rewards achieved by DATE-V and the other 4 benchmarks. As expected, Oracle achieves the highest reward which gives an upper bound to the other algorithms. Among the others, we see that the proposed algorithm significantly outperforms other benchmarks by taking into account the context of tasks and vehicles. It can be observed in the figure that the cumulative reward of DATE-V is similar to that of the Random algorithm in the first 2,500 tasks. This is because the RSU does not enough knowledge at the beginning to link the replications' context and qualities, therefore it randomly explores the available replications,  which is exactly the same as the Random algorithm. After a period of exploration, the proposed algorithm is able to exploit the learned knowledge and we see that cumulative reward of the proposed algorithm begins to approach the cumulative reward of Oracle. For the UCB algorithm, its cumulative reward is similar to that of the Random algorithm. The malfunction of UCB is mainly due to the large the arm sets (TaV-SeV pairs) and hence the UCB algorithm is stuck in the exploration. Further analyzing the rewards achieve by mLinUCB, we know that considering the context for each possible replication is not effective to produce a good result due to the large arm set. We also show the average reward for each replication in Fig. \ref{fig:ave_reward}. We see that the average reward for a task replication achieved by Oracle stabilizes at around 0.25 and DATE-V increases the average replication reward from 0.14 to 0.23. This means that DATE-V can learn context-specific replication qualities over time and after sufficiently many tasks, it selects replications almost as well as Oracle does.

\subsection{DATE-V with Delayed Feedback}
Fig. \ref{fig:vary_lambda} shows the cumulative rewards achieved by DATE-V with non-delayed/delayed feedback. In general, we see that delayed feedback does not incur significant performance loss. We also evaluate the effect of task arrival rate in the delayed feedback scenario. The simulation result is consistent with the theoretic analysis in Proposition \ref{theo:upper_bound_DF} that a higher task arrival rate leads to a larger regret.

\subsection{Impact of Budget}
Fig. \ref{fig:vary_b} depicts the cumulative reward achieved by 5 algorithms under different budgets. We see that the cumulative rewards achieved by DATE-V and Oracle grow with the increase in budget since more beneficial replications can be selected to maximize the reward. It is worth noticing that the cumulative rewards become saturated when the budget is larger than 3 since the proposed algorithm considers the submodularity of the reward function and therefore stops smartly when the marginal reward is low. By contrast, UCB, mLinUCB, and Random always utilize the full budget and select replications without considering the submodular reward. Therefore, with a larger budget, these algorithms keep adding replications when marginal rewards become negative, which decreases the cumulative reward.
\begin{figure} [tb]
	\centering
	\includegraphics[width=0.5\linewidth]{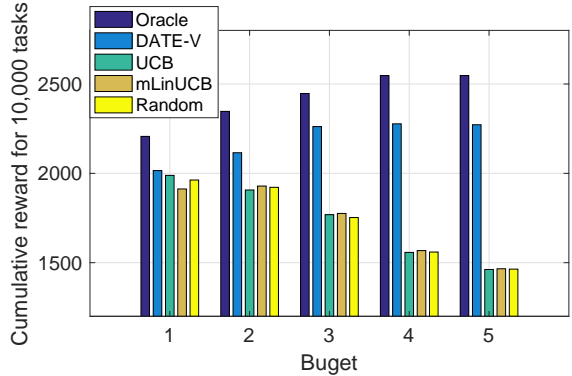}
	\caption{Impact of budget}
	\label{fig:vary_b}
	\vspace{-0.15 in}
\end{figure}%

\subsection{Impact of Task Deadline}
Fig. \ref{fig:vary_L} shows the cumulative rewards achieved by Oracle, DATE-V, and Random with different task deadlines. We see that the cumulative rewards achieved by all three algorithms grow with the increase in the mean task deadline. The reason for this trend is intuitive: the tasks are more likely to be completed if the deadline is loose. In addition, the gap of cumulative reward between DATE-V and Random diminishes as the deadline increases. This is because most of the replications can be completed with loose deadlines and hence the benefit of smart replication provided by DATE-V decreases.
\begin{figure} [tb]
	\centering
	\includegraphics[width=0.5\linewidth]{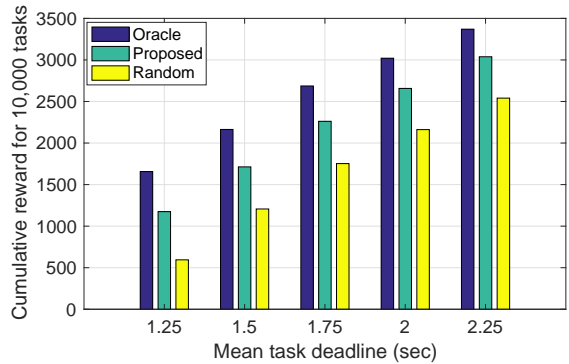}
	\caption{Impact of task deadline.}
	\label{fig:vary_L}
	\vspace{-0.15 in}
\end{figure}%

\section{Conclusion}\label{sec:conclusion}
In this paper, we investigated the task replication for deadline-constrained tasks in VCC systems. A RSU-assisted task scheduling framework is constructed, and a novel task replication algorithm, called DATE-V, is proposed to guarantee the timeliness of task processing. DATE-V addresses many key concerns in VCC systems. It uses side-information (context) of tasks, vehicles to learn the completion probability of a replication with the uncertainty in vehicle movements. The combinatorial feature of DATE-V allows multiple replications to be made for each task such that the completion probability of a task is increased. The DATE-V is practical, easy to implement and scalable to large vehicular networks while achieving provably asymptotically optimal performance. Besides the task replication for VCC, our framework can also be applied to many other sequential decision making problems under uncertainty that involve multiple-play given a limited budget and context information.

\bibliographystyle{IEEEtran}
\bibliography{refs}

\newpage
\appendices
\section{Proof of Proposition \ref{prop:greedy_opt}} \label{proof:prop:greedy_op}
\begin{proof}
	The optimality of greedy algorithm is due to the unique property of submodular reward function \eqref{eq:reward_func} indicated in the following Lemma.
	\begin{lemma}\label{lemma}
		For any two replication $v,v^\prime \in \mathcal{V}^t\backslash \mathcal{A}^t$, if $\mu_{v^\prime} \geq \mu_v$ holds true, we have $\Delta(\bm{\mu}, \{v^\prime\} | \mathcal{A}^t) \geq \Delta(\bm{\mu}, \{v\} | \mathcal{A}^t)$.
	\end{lemma}
	\begin{proof}
		This property can be easily verified by definition of reward function and marginal reward. 
	\end{proof}
	We now consider a special case where the number of replications to select is fixed in advance, i.e. $|\mathcal{A}^t| = b$ where $b$ is a constant. In this case, $\eta\cdot c |\mathcal{A}^t|$ is a constant and the solution to problem \textbf{P2} is to select $b$ replications that have the highest expected quality:
	\begin{align}
	\tilde{v}_k = \argmax_{v\in\mathcal{V}^t\backslash\{\cup^{k-1}_{i=1}\tilde{v}_i\}} \mu_v,~~~k = 1,2,\dots,b
	\end{align}
	In addition, the SeVs selected by the greedy algorithm can be rewritten as: 
	\begin{align}
	v_k = \argmax_{v\in\mathcal{V}^t\backslash\{\bigcup^{k-1}_{i=1}v_i\}} \Delta(\bm{\mu}, \{v\}|\cup^{k-1}_{i=1}v_i), ~~~k = 1,2,\dots,b
	\end{align}
	From Lemma \ref{lemma}, we know that the two sequences $\{\tilde{v}_k\}^b_{k=1}$ and $\{v_k\}^b_{k=1}$ are identical. 
	Next, we need the determine the number of replications to maximize the reward. Note that the reward function can be written as a sum of marginal rewards: $u(\bm{\mu},\cup^{b}_{k=1}v_k) = \sum_{k=1}^{b} \Delta(\bm{\mu},\{v_k\}|\mathcal{A}_{k-1})$. Moreover, by following the algorithm design, we will have $\Delta(\bm{\mu},\{v_i\} | \mathcal{A}_{i-1})\geq\Delta(\bm{\mu},\{v_j\} | \mathcal{A}_{j-1}), \forall i < j$. Therefore, to maximize the reward, the greedy algorithm should stop at the $k$-th iteration if $\Delta(\bm{\mu},\{v_k\}|\mathcal{A}_{k-1})\leq 0$. We now can conclude that the greedy algorithm is able to achieve the optimal solution for the problem \textbf{P2}.
\end{proof}

\section{Proof of Proposition \ref{theo:upper_bound}}\label{proof:theo:upper_bound}
Before proceeding, we first define some auxiliary variables. For each hypercube $p\in\mathcal{P}_{T}$, we define $\bar{\mu}(p)=\sup_{\phi \in p}\mu(\phi)$ and $\ubar{\mu}(p)=\inf_{\phi \in p}\mu(\phi)$ be the best and worst expected quality over all contexts $\phi \in p$. In some steps of the proofs, we need to compare the qualities at different positions in a hypercube. As a point of reference, we define the context at (geometrical) center of a hypercube $p$ as $\tilde{\phi}_p$ and its expected quality $\tilde{\mu}(p) = \mu(\tilde{\phi}_p)$. Given the replication set $\mathcal{V}^t = \{1,2,\dots, V^t\}$, context set $\bm{\phi}^t = \{\phi^t_1,\phi^t_2, \dots, \phi^t_{V^t}\}$, and the hypercube set $\p^t = \{p^t_1,p^t_2, \dots, p^t_{V^t}\}$ for each task $t$, let 
\begin{align*}
\bar{\bm{\mu}}^t = [\bar{\mu}(p^t_1),\dots,\bar{\mu}(p^t_{V^t})],\\
\ubar{\bm{\mu}}^t = [\ubar{\mu}(p^t_1),\dots,\ubar{\mu}(p^t_{V^t})],\\
\tilde{\bm{\mu}}^t = [\tilde{\mu}(p^t_1),\dots,\tilde{\mu}(p^t_{V^t})].
\end{align*}
In addition, for a task $t$, we define a replication set $\tilde{\mathcal{A}}^{t}$ satisfying:
\begin{align}\label{eq:tilde_optimal}
\tilde{\mathcal{A}}^{t} = \argmax_{\mathcal{A} \subseteq \mathcal{V}^t, |\mathcal{A}| \leq b^t} u(\tilde{\bm{\mu}}^t, \mathcal{A}) 
\end{align}
The replication set $\tilde{\mathcal{A}}^{t}$ is used to identify subsets of replications which are bad to select. Let 
\begin{align}
\mathcal{L}^t=\left\{ G \subseteq \mathcal{V}^t, |G| \leq b^t: u(\ubar{\bm{\mu}}^t, \tilde{\mathcal{A}}^t ) - u(\bar{\bm{\mu}}^t, G ) \geq At^\theta \right\}
\end{align} 
be the \emph{set of suboptimal subsets of arms} for hypercube set $\p^t$, where $A>0$ and $\theta<0$ are parameters used only in the regret analysis. We call a subset $G$ of replication in $\mathcal{L}^t$ is \emph{suboptimal} for $\p^t$, since the sum of the worst expected reward in $\tilde{\mathcal{A}}^{t}$ is at least an amount $At^\theta$ higher than the sum of the best expected reward for subset $G$. We call subsets in $\mathcal{A}^t_{b^-}\backslash\mathcal{L}^t$ \emph{near-optimal} for $\p^t$. Here, $\mathcal{A}^t_{b^-}$ denotes the set of all element subsets with element number less than $b^t$. Then, the expected regret $R(T)$ can be divided into the following three summands:
\begin{align}
R(T) = \mathbb{E}[R_e(T)]+\mathbb{E}[R_s(T)]+\mathbb{E}[R_n(T)]
\end{align}
where the term $\mathbb{E}[R_e(T)]$ is the regret due to exploration phases and the terms $\mathbb{E}[R_s(T)]$ and $\mathbb{E}[R_n(T)]$ both correspond to regret in exploitation phases: $\mathbb{E}[R_s(T)]$ is the regret due to suboptimal choices, i.e., the subsets of replications from $\mathcal{L}^t$ are selected; $\mathbb{E}[R_n(T)]$ is the regret due to near-optimal choices, i.e., the subsets of replications from $\mathcal{A}^t_{b^-}\backslash\mathcal{L}^t$ are selected. In the following, we prove that each of the three summands is bounded.

We first give a bound of $\mathbb{E}[R_e(T)]$, which depends on the choice of two parameters $z$ and $\gamma$.

\begin{lemma}[Bound for $\mathbb{E}(R_e(T))$] \label{lemma:bound_R_e}
	Let $K(t)=t^{z}\log(t)$ and $h_{T}=\lceil T^{\gamma} \rceil$, where $0<z<1$ and $0<\gamma<\frac{1}{D}$. If the algorithm is run with these parameters, the regret $\mathbb{E}[R_e(T)]$ is bounded by 
	\begin{align}
	\mathbb{E}[R_e(T)]\leq (1+\eta B)2^D\left(T^{z+\gamma D}\log(T) + T^{\gamma D}\right)
	\end{align}
\end{lemma}
\begin{proof}[Proof of Lemma \ref{lemma:bound_R_e}]
	Suppose the algorithm enters the exploration phase for task $t$ and let $\p^t=(p^t_{v})_{v\in\mathcal{V}^t}$ be the hypercubes of currently available replications. Then, based on the design of DATE-V, the set of under-explored hypercubes $\mathcal{P}^{\text{ue},t}_T$ is non-empty, i.e., there exists at least one replication with context $\phi^t_v$, such that a hypercube $p$ satisfying $\phi^t_v \in p$ has $C^t(p)\leq K(t)=t^{z}\log(t)$. Clearly, there can be at most $\lceil T^{z} \log(T) \rceil$ exploration phases in which replications in $p$ are selected due to under-exploration of $p$. Since there are $(h_T)^{D}$ hypercubes in the partition, there can be at most $(h_{T})^{D}\lceil T^{z} \log(T) \rceil$ exploration phases. Notice that the maximum achievable reward of a replication decision $\mathcal{A}^t$ is bounded by $(1-\eta)$ and the minimum achievable reward is $-B\eta$, where $B$ is the maximum possible budget for a task. The maximum regret of wrong replication selection in one exploration phase is bounded by $1+\eta(B-1) < 1+\eta B$. Therefore, we have
	\begin{align*}
	\mathbb{E}[R_e(T)] & \leq  (1+\eta B) (h_{T})^{D}\lceil T^{z}\log(T)\rceil\\
	& =  (1+\eta B) \lceil T^\gamma\rceil^{D}\lceil T^{z}\log(T)\rceil
	\end{align*} 
	Using $\lceil T^{\gamma}\rceil^{D} \leq (2T^{\gamma})^{D} =2^{D}T^{\gamma D}$, it holds 
	\begin{align}
	\mathbb{E}[R_e(T)] & \leq (1+\eta B)2^{D}T^{\gamma D}\left(T^z\log(T) +1\right)\nonumber\\
	& = (1+\eta B)2^D\left(T^{z+\gamma D}\log(T) + T^{\gamma D}\right)
	\end{align}
\end{proof}
Next, we give a bound for $\mathbb{E}[R_s(T)]$. This bound also depends on the choice of two parameters $z$ and $\gamma$. Additionally, a condition on these parameters has to be satisfied.

\begin{lemma}[Bound for  $\mathbb{E}(R_s(T))$] \label{lemma:bound_R_s}
	Let $K(t)=t^{z}\log(t)$ and $h_{T}=\lceil T^{\gamma} \rceil$, where $0<z<1$ and $0<\gamma<\frac{1}{D}$. If the algorithm is run with these parameters, Assumption \ref{holder} holds true and the additional condition $2H(t)+ 2BLD^{\frac{\alpha}{2}}h_T^{-\alpha}\leq At^\theta$ is satisfied for all $1\leq t\leq T$ where $H(t) = Bt^{-\frac{z}{2}}$, the regret $\mathbb{E}[R_s(T)]$ is bounded by 
	\begin{align}
	\mathbb{E}[R_s(T)] \leq (1+\eta B)\left(\sum_{k=1}^B{V^{\max} \choose k}\right) \frac{\pi^2}{3}
	\end{align}
\end{lemma}
\begin{proof}[Proof of Lemma \ref{lemma:bound_R_s}]
	For $1 \leq t \leq T$, let $W^t=\{\mathcal{P}^{\text{ue},t}=\emptyset\}$ be the even that slot $t$ is an exploitation phase. By the definition of $\mathcal{P}^{\text{ue},t}$, in this case, it holds that $C^t(p^t_v)>K(t) = t^{z}\log(t), \forall p^t_v \in \p^t$. Let $V_G^t$ be the event that subset $G\in\mathcal{L}^t$ is selected at for task $t$. Then, it holds that 
	\begin{align}
	R_s(T)= & \sum_{t=1}^{T}\sum_{G\in\mathcal{L}^t(\p^t)} I_{\{V_G^t, W^t\}} \times \left(r\left(\mathcal{A}^{*,t}\right)-r\left(G\right)\right)
	\end{align}
	where, for each task, the loss due to selecting a suboptimal subset $G\in\mathcal{L}^t$ is considered. Since the maximum regret of selecting $G$ is bounded by $(1+\eta B)$, we have
	\begin{align}
	R_s(T)\leq (1+\eta B)\sum_{t=1}^{T}\sum_{G\in\mathcal{L}^t} I_{\{V_G^t, W^t\}},
	\end{align}
	and taking the exception, the regret is hence bounded by
	\begin{align}
	\mathbb{E}[R_s(T)] &\leq (1+\eta B)\sum_{t=1}^{T}\sum_{G\in\mathcal{L}^t} \mathbb{E}\left[I_{\{V_G^t, W^t\}}\right] \nonumber\\
	& = (1+\eta B)\sum_{t=1}^{T}\sum_{G\in\mathcal{L}^t} \Pr \left\{V_G^t, W^t\right\}
	\end{align}
	
	In the event of $V^t_G$, by the design of the algorithm, this means that with the estimated replication quality, the rewards of selecting replications in $G$ is at least as high as the reward of selecting replications in $\tilde{\mathcal{A}}^t$, i.e., $u(\hat{\bm{\mu}}^t,G) \geq u(\hat{\bm{\mu}}^t,\tilde{\mathcal{A}}^t)$. Thus, we have:
	\begin{align} \label{prob_subopt}
	\Pr\left\{V_G^t,W^t\right\} \leq \Pr\left\{u(\hat{\bm{\mu}}^t,G) \geq u(\hat{\bm{\mu}}^t,\tilde{\mathcal{A}}^t)\right\}
	\end{align}
	
	The event in the right-hand side of \eqref{prob_subopt} implies at lease one of the three following events for any $H(t)>0$:
	\begin{align*}
	E_1= \left\{u(\hat{\bm{\mu}}^t,G) \geq u(\bar{\bm{\mu}}^t,G)+H(t), W^t\right\}
	\end{align*}
	\begin{align*}
	E_2= \left\{u(\hat{\bm{\mu}}^t,\tilde{\mathcal{A}}^t) \leq u(\ubar{\bm{\mu}}^t,\tilde{\mathcal{A}}^t)-H(t), W^t\right\}
	\end{align*}
	\begin{align*}
	E_3 = & \left\{u(\hat{\bm{\mu}}^t,G) \geq u(\hat{\bm{\mu}}^t,\tilde{\mathcal{A}}^t), u(\hat{\bm{\mu}}^t,G) < u(\bar{\bm{\mu}}^t,G)+H(t), \right.\\
	&\quad \left.u(\hat{\bm{\mu}}^t,\tilde{\mathcal{A}}^t) > u(\ubar{\bm{\mu}}^t,\tilde{\mathcal{A}}^t)-H(t), W^t\right\}.
	\end{align*}
	
	Hence, we have for the original event in \eqref{prob_subopt}
	\begin{align}\label{ori_E123}
	\left\{u(\hat{\bm{\mu}}^t,G) \geq u(\hat{\bm{\mu}}^t,\tilde{\mathcal{A}}^t)\right\}\subseteq E_1 \cup E_2 \cup E_3
	\end{align}
	
	The probability of the three event $E_1$, $E_2$, and $E_3$ will be bounded separately. Let start by bounding $E_1$. Recall that the best expected quality of replications in set $p$ is $\bar{\mu}(p)=\sup_{\phi \in p} \bar{\mu}(\phi)$. Therefore, the expected quality of replication $v$ in $G$ is bounded by
	\begin{align}
	\mathbb{E}\left[\hat{\mu}(p^t_v)\right]= & \mathbb{E}\left[\frac{1}{C^t(p^t_v)} \sum_{(\tau,k): \phi^\tau_k\in p^\tau_v, k\in\mathcal{A}^\tau} q(\phi^\tau_k)\right] \nonumber\\
	= & \frac{1}{C^t(p^t_v)|} \underbrace{\sum_{(\tau,k): \phi^\tau_k\in p^\tau_v, k\in\mathcal{A}^\tau}}_{C^t(p^t_v) \text{summands}} \underbrace{\mu(\phi^\tau_k)}_{\leq \bar{\mu}(p^t_v)} \nonumber\\
	\leq &  \bar{\mu}(p^t_v)
	\end{align}
	
	This implies 
	\begin{align}
	\text{Prob}\{E_1\}\nonumber &=\Pr\left\{u(\hat{\bm{\mu}}^t,G) \geq  u(\bar{\bm{\mu}}^t,G) + H(t), W^t\right\}\nonumber\\
	&\leq \Pr\left\{\hat{\mu}(p^t_v) \geq \bar{\mu}(p^t_v) + \frac{H(t)}{B}, \exists v\in G,  W^t\right\}\nonumber\\
	&\leq \Pr\left\{\hat{\mu}(p^t_v) \geq \mathbb{E}\left[\hat{\mu}(p^t_v)\right] + \frac{H(t)}{B}, \exists v\in G,  W^t\right\}\nonumber\\
	& = \sum_{v\in G}\Pr\left\{\hat{\mu}(p^t_v) \geq  \mathbb{E}\left[\hat{\mu}(p^t_v)\right]+ \frac{H(t)}{B}, W^t\right\}\nonumber
	\end{align}
	where $ B = \max\limits_{t=1,\dots,T} b^t$ the maximum budget a replication could have. The first inequality comes from the fact that
	\begin{align*}
	&\left\{G \subseteq \mathcal{V}^t~|~u(\hat{\bm{\mu}}^t,G) \geq  u(\bar{\bm{\mu}}^t,G) + H(t)\right\} \\ \subseteq &\left\{G \subseteq \mathcal{V}^t~|~\hat{\mu}(p^t_v) \geq \bar{\mu}(p^t_v) + \frac{H(t)}{B}, \exists v\in G\right\}
	\end{align*}
	which can be easily verified by \emph{reductio ad absurdum} with the expected reward function. Now, applying Chernoff-Hoeffding bound \cite{hoeffding1963probability} (note that for each replication, the estimated quality is bounded by $1$ and exploiting that event $W^t$ implies that at least $t^{z}\log(t)$ samples were drawn, we get
	\begin{align}\label{eq:prob_E_1}
	\Pr\{E_1\} &\leq \sum_{v\in G} \text{Prob}\left\{\hat{\mu}(p^t_v) -\mathbb{E}\left[\hat{\mu}(p^t_v)\right] \geq \frac{H(t)}{B}, W(t)\right\}\nonumber\\
	&\leq \sum_{v\in G} \exp\left(\dfrac{-2C^t(p^t_v)H(t)^2}{B^2}\right)\nonumber\\
	&\leq \sum_{v\in G} \exp\left(\dfrac{-2H(t)^2t^{z}\log(t)}{B^2}\right)
	\end{align}
	
	Analogously, it can be proven for event $E_2$, that
	\begin{align}\label{eq:prob_E_2}
	\Pr\{E_2\} &= \text{Prob}\left\{u(\hat{\bm{\mu}}^t,\tilde{\mathcal{A}}^t) \geq u(\ubar{\bm{\mu}}^t,\tilde{\mathcal{A}}^t)-H(t), W^t\right\}\nonumber \\
	&\leq \sum_{v\in \mathcal{\tilde{\mathcal{A}}}^{t}} \exp\left(\dfrac{-2H(t)^2t^{z}\log(t)}{B^2}\right)
	\end{align}
	
	To bound the event $E_3$, we first make some additional definitions. First, we rewrite the estimate $\hat{\mu}(p), p\in \mathcal{P}_T$ as follows:
	\begin{align*}
	\hat{\mu}(p)&=\dfrac{1}{C^t(p)}\sum_{(\tau,k):\phi^\tau_k\in p, k \in \mathcal{A}^\tau} q(\phi^\tau_k)\\
	&= \dfrac{1}{C^t(p)}\sum_{(\tau,k):\phi^\tau_k\in p, k\in\mathcal{A}^\tau} \mu(\phi^\tau_{k})+\epsilon^\tau_{k}
	\end{align*}
	where $\epsilon^\tau_{k}$ denotes the deviation from the expected quality of a replication $k \in \mathcal{A}^\tau$ with context $\phi^\tau_{k}$. Additionally, we define the best and worst context for a hypercube $p\in\mathcal{P}_{T}$ as $\phi^{\text{best}}(p) \triangleq \argmax_{\phi \in p} \mu(\phi)$ and $\phi^{\text{worst}}(p) \triangleq \argmin_{ \phi \in p} \mu(\phi)$, respectively. Finally, we define the best and worst quality of a replication in hypercube $p$ as 
	\begin{align}
	\mu^{\text{best}}(p)= \dfrac{1}{C^t(p)}\sum_{(\tau,k):\phi^\tau_k \in p, k \in \mathcal{A}^\tau} \mu(\phi^{\text{best}}(p))+\epsilon^\tau_{k} \label{d_best}\\
	\mu^{\text{worst}}(p)= \dfrac{1}{C^t(p)}\sum_{(\tau,k):\phi^\tau_k\in p, k \in \mathcal{A}^\tau} \mu(\phi^{\text{worst}}(p))+\epsilon^\tau_{k} \label{d_worst}
	\end{align}
	Let $\bm{\mu}^{\text{best},t} = \left[\mu^{\text{best}}(p^t_1), \mu^{\text{best}}(p^t_2),\dots, \mu^{\text{best}}(p^t_{V^t})\right]$ and $\bm{\mu}^{\text{worst},t} = \left[\mu^{\text{worst}}(p^t_1), \mu^{\text{worst}}(p^t_2),\dots, \mu^{\text{worst}}(p^t_{V^t})\right]$.
	
	By H\"{o}lder condition from Assumption \ref{holder}, since $\phi^{\text{best}}(p) \in p$ and only contexts from hypercube $p$
	are used for calculating the estimated quality $\hat{\mu}(p)$, it can be shown that 
	\begin{align}\label{holder_best}
	\mu^{\text{best}}(p)-\hat{\mu}(p) \leq LD^{\frac{\alpha}{2}}h^{-\alpha}_{T}
	\end{align}
	holds. Analogously, we have 
	\begin{align}\label{holder_worst}
	\hat{\mu}(p)-\mu^{\text{worst}}(p) \leq LD^{\frac{\alpha}{2}}h^{-\alpha}_{T}
	\end{align}
	Applying \eqref{holder_best} and \eqref{holder_worst} to replication in $G$ and $\tilde{\mathcal{A}}^{t}(\p^t)$, we have:
	\begin{align}
	u(\bm{\mu}^{\text{best},t},G) - u(\hat{\bm{\mu}}^t,G) \leq & \sum_{v \in G} \left(\mu^{\text{best}}(p^t_v)-\hat{\mu}(p^t_v)\right) \nonumber \\ \leq &BLD^{\frac{\alpha}{2}}h^{-\alpha}_{T} \label{eq:G_best_esti} 
	\end{align}
	\begin{align}
	u(\hat{\bm{\mu}}^t,\tilde{\mathcal{A}}^{t}) - u(\bm{\mu}^{\text{worst},t},\tilde{\mathcal{A}}^{t}) \leq & \sum_{v \in \tilde{\mathcal{A}}^{t}}\left(\hat{\mu}(p^t_v)-\mu^{\text{worst}}(p^t_v)\right) \nonumber\\ \leq &BLD^{\frac{\alpha}{2}}h^{-\alpha}_{T}\label{eq:S*_best_esti}
	\end{align}
	Now the three components of event $E_3$ are considered separately. By definition of $\mu^{\text{best}}(p)$ and $\mu^{\text{worst}}(p)$ in \eqref{d_best} and \eqref{d_worst}. The first component of $E_3$ holds that
	\begin{align}\label{E_3_1}
	\left\{u(\hat{\bm{\mu}}^t,G) \geq  u(\hat{\bm{\mu}}^t,\tilde{\mathcal{A}}^{t})\right\}
	\subseteq \left\{u(\bm{\mu}^{\text{best},t},G) \geq u(\bm{\mu}^{\text{worst},t},\tilde{\mathcal{A}}^{t})\right\}
	\end{align}
	
	For the second component, using \eqref{eq:G_best_esti}, we have
	\begin{align}\label{E_3_2}
	&\left\{u(\hat{\bm{\mu}}^t,G) < u(\bar{\bm{\mu}}^t,G)+H(t)\right\} \nonumber\\
	\subseteq & \left\{u(\bm{\mu}^{\text{best},t},G) - BLD^{\frac{\alpha}{2}}h^{-\alpha}_T < u(\bar{\bm{\mu}}^t,G) + H(t)\right\} \nonumber\\
	= & \left\{u(\bm{\mu}^{\text{best},t},G) < u(\bar{\bm{\mu}}^t,G) + BLD^{\frac{\alpha}{2}}h^{-\alpha}_T + H(t)\right\}
	\end{align}
	
	For the third component, we have 
	\begin{align}\label{E_3_3}
	&\left\{u(\hat{\bm{\mu}}^t,\tilde{\mathcal{A}}^t) > u(\ubar{\bm{\mu}}^t,\tilde{\mathcal{A}}^t)-H(t)\right\} \nonumber\\
	\subseteq & \left\{u(\bm{\mu}^{\text{worst},t},\tilde{\mathcal{A}}^{t})+BLD^{\frac{\alpha}{2}}h^{-\alpha}_{T} > u(\ubar{\bm{\mu}}^t,\tilde{\mathcal{A}}^t) - H(t)\right\} \nonumber\\
	= & \left\{u(\bm{\mu}^{\text{worst},t},\tilde{\mathcal{A}}^{t}) > u(\ubar{\bm{\mu}}^t,\tilde{\mathcal{A}}^t) - BLD^{\frac{\alpha}{2}}h^{-\alpha}_{T} - H(t)\right\}
	\end{align}
	
	Therefore, using \eqref{E_3_1}, \eqref{E_3_2} and \eqref{E_3_3}, the probability of event $E_3$ is bounded by 
	\begin{align}\label{E_3_prob_bound}
	~\Pr&\{E_3\}\\
	\leq~\Pr&\left\{W^t, u(\bm{\mu}^{\text{best},t},G) \geq u(\bm{\mu}^{\text{worst},t},\tilde{\mathcal{A}}^{t}), \right.\nonumber\\
	&~u(\bm{\mu}^{\text{best},t},G) < u(\bar{\bm{\mu}}^t,G) + BLD^{\frac{\alpha}{2}}h^{-\alpha}_T + H(t) \nonumber\\
	&~\left. u(\bm{\mu}^{\text{worst},t},\tilde{\mathcal{A}}^{t}) > u(\ubar{\bm{\mu}}^t,\tilde{\mathcal{A}}^t) - BLD^{\frac{\alpha}{2}}h^{-\alpha}_{T} - H(t) \right\}.\nonumber
	\end{align}
	
	We want to find a condition under which the probability for $E_3$ is zero. For this purpose, it is sufficient to show that the probability for the right-hand side in \eqref{E_3_prob_bound} is zero. Suppose that the following condition is satisfied:
	\begin{align}\label{eq:condition_E3}
	2H(t)+2BLD^{\frac{\alpha}{2}}h^{-\alpha}_{T} \leq At^\theta
	\end{align}
	
	Since $G\in\mathcal{L}^t$, we have $u(\ubar{\bm{\mu}}^t, \tilde{\mathcal{A}}^t) - u(\bar{\bm{\mu}}^t, G) \geq A t^\theta$, which together with \eqref{eq:condition_E3} implies that
	\begin{align}
	u(\ubar{\bm{\mu}}^t, \tilde{\mathcal{A}}^t) - u(\bar{\bm{\mu}}^t, G) -\left(2H(t)+2BLD^{\frac{\alpha}{2}}h^{-\alpha}_{T}\right)\geq 0
	\end{align}
	Rewriting yields
	\begin{align}\label{condition_rewr}
	& u(\ubar{\bm{\mu}}^t, \tilde{\mathcal{A}}^t) - H(t)-BLD^{\frac{\alpha}{2}}h^{-\alpha}_{T} \nonumber \\ \geq & u(\bar{\bm{\mu}}^t, G) + H(t)+BLD^{\frac{\alpha}{2}}h^{-\alpha}_{T}
	\end{align}
	
	If \eqref{condition_rewr} holds true, the three components of the right-hand side in \eqref{E_3_prob_bound} cannot be satisfied at the same time: Combining the second and third component of \eqref{E_3_prob_bound} with \eqref{condition_rewr} yields $u(\bm{\mu}^{\text{best},t},G) < u(\bm{\mu}^{\text{worst},t},\tilde{\mathcal{A}}^{t})$, which contradicts the first term of \eqref{E_3_prob_bound}. Therefore, under condition \eqref{eq:condition_E3}, it follows that $\Pr\{E_3\}=0$.
	
	So far, the analysis was performed with respected to an arbitrary $H(t)>0$. In the remainder of the proof, we choose $H(t)=Bt^{-z/2}$. Then, using \eqref{eq:prob_E_1} and \eqref{eq:prob_E_2}, we have
	\begin{align}
	\text{Prob}\{E_1\} \leq & B\exp\left(\dfrac{-2H(t)^2t^{z}\log(t)}{B^2}\right) \nonumber \\
	\leq & B\exp\left(-2\log(t)\right) \nonumber \\
	\leq & Bt^{-2}
	\end{align}
	and analogously
	\begin{align}
	\text{Prob}\{E_2\} \leq Bt^{-2}
	\end{align}
	
	To sum up, under condition \eqref{eq:condition_E3}, using \eqref{ori_E123}, the probability in \eqref{prob_subopt} is bounded by 
	\begin{align*}
	\Pr\left\{V_G^t,W^t\right\}	\leq &\Pr\left\{E_1\cup E_2\cup E_3\right\} \nonumber\\
	\leq  &\Pr\left\{E_1\right\} + \Pr\left\{E_2\right\} + \Pr\left\{E_3\right\} \nonumber\\ 
	\leq  &2Bt^{-2}
	\end{align*}
	Given this we have:
	\begin{align}
	\mathbb{E}[R_s(T)] \leq & (1+\eta B)\times\sum_{t=1}^{T}\sum_{G\in\mathcal{L}^t} \Pr\left\{V_G^t, W^t\right\} \nonumber \\
	\leq & (1+\eta B)|\mathcal{L}^t|\sum_{t=1}^{T}2Bt^{-2} \nonumber \\
	\leq &  (1+\eta B)|\mathcal{L}^t| 2\sum_{t=1}^{\infty}t^{-2} \nonumber \\
	\leq &  (1+\eta B)|\mathcal{L}^t| \frac{\pi^2}{3}\nonumber\\
	\leq &  (1+\eta B)\left(\sum_{k=1}^B{V^{\max} \choose k}\right) \frac{\pi^2}{3}
	\end{align}
	where $\sum_{k=1}^B{V^{\max} \choose k}$ is maximum possible number of subsets with size less than or equal $B$ where $V^{\max}$ is the .
\end{proof}

Now we give a bound for $\mathbb{E}\left[R_n(T)\right]$.
\begin{lemma}[Bound for  $\mathbb{E}(R_n(T))$] \label{lemma:bound_R_n}
	Let $K(t)=t^{z}\log(t)$ and $h_{T}=\lceil T^{\gamma} \rceil$, where $0<z_n<1$ and $0<\gamma<\frac{1}{D}$. If the algorithm is run with these parameters, Assumption \ref{holder} holds true, the regret $\mathbb{E}[R_n(T)]$ is bounded by 
	\begin{align}
	\mathbb{E}[R_n(T)]\leq 3BLD^{\frac{\alpha}{2}}T^{1-\gamma\alpha}+\dfrac{A}{1+\theta}T^{1+\theta}
	\end{align}
	
\end{lemma}
\begin{proof}[Proof of Lemma \ref{lemma:bound_R_n}]
	For $1 \leq t \leq T$, consider the event $W^t$ as in the previous proof, the regret due to near-optimal subsets can be written as
	\begin{align}
	& R_n(T)=\sum_{t=1}^{T}I_{\{W^t,\mathcal{A}^t \in \{\mathcal{A}^t_{b^-}\backslash\mathcal{L}^t\}\}}\times\left( r\left(\mathcal{A}^{*,t}\right)-r\left(\mathcal{A}^t\right)\right)
	\end{align}
	where in each time slot in which the selected subset $\mathcal{A}^t$ is near-optimal, i.e., $\mathcal{A}^t\in\mathcal{A}_{b^-} \backslash\mathcal{L}^t$, the regret is considered for selecting $\mathcal{A}^t$ instead of the $\mathcal{A}^{*,t}$. Let $Q^t =W^t\cap\{\mathcal{A}^t\in\mathcal{A}_{b^-} \backslash \mathcal{L}^t\}$ denotes the event of selecting a near-optimal arm set. Then, we have
	\begin{align*}
	\mathbb{E}\left[R_n(T)\right]=&\sum_{t=1}^{T}\Pr\{Q(t)\} \mathbb{E}\left[r\left(\mathcal{A}^{*,t}\right)-r\left(\mathcal{A}^t\right) \mid Q(t)\right]\\
	\leq &\sum_{t=1}^{T} u\left(\bm{\mu}^t,\mathcal{A}^{*,t}\right)-u\left(\bm{\mu}^t,\mathcal{A}^t\right) 
	\end{align*}
	
	Now, let $t$ be the time slot, where $Q(t)$ holds true, i.e., the algorithm enters an exploitation phase and $J\in\mathcal{A}_{b^-}\backslash\mathcal{L}^t$. By the definition of $\mathcal{P}^{\text{ue},t}$, it holds that $C^t(p^t_{v})>K(t)=t^{z}\log(t)$ for all $p^t_v\in\p^t$. In addition, since $J\in\mathcal{A}_{b^-} \backslash\mathcal{L}^t$, it holds
	\begin{align}
	u(\ubar{\bm{\mu}}^t, \tilde{\mathcal{A}}^t) - u(\bar{\bm{\mu}}^t, J) < At^\theta
	\end{align}
	To bound the regret, we have to give an upper bound on
	\begin{align}
	\sum_{t=1}^{T}\left(u\left(\bm{\mu}^t,\mathcal{A}^{*,t}\right)-u\left(\bm{\mu}^t,J\right)\right)
	\end{align}
	
	Applying H\"{o}lder condition several times yields:
	\begin{align*}
	& u\left(\bm{\mu}^t,\mathcal{A}^{*,t}\right)-u\left(\bm{\mu}^t,J\right)\\
	\leq & u\left(\tilde{\bm{\mu}}^t,\mathcal{A}^{*,t}\right) + BLD^{\frac{\alpha}{2}}h^{-\alpha}_{T} -u\left(\bm{\mu}^t,J\right)\\
	\leq & u\left(\tilde{\bm{\mu}}^t,\tilde{\mathcal{A}}^{t}\right) + BLD^{\frac{\alpha}{2}}h^{-\alpha}_{T} -u\left(\bm{\mu}^t,J\right)\\
	\leq & u\left(\ubar{\bm{\mu}}^t,\tilde{\mathcal{A}}^t\right) + 2BLD^{\frac{\alpha}{2}}h^{-\alpha}_{T} -u\left(\bm{\mu}^t,J\right)\\
	\leq & u\left(\ubar{\bm{\mu}}^t,\tilde{\mathcal{A}}^t\right) + 3BLD^{\frac{\alpha}{2}}h^{-\alpha}_{T} -u\left(\bar{\bm{\mu}}^t,J\right)\\
	\leq & 3BLD^{\frac{\alpha}{2}}h^{-\alpha}_{T}+At^{\theta}
	\end{align*}
	where the third inequality follows the definition of $\tilde{\mathcal{A}}^{t}$. Using $h^{-\alpha}_{T}=\lceil T^{\gamma} \rceil^{-\alpha} \leq T^{-\gamma\alpha}$, we further have
	\begin{align}
	u\left(\bm{\mu}^t,\mathcal{A}^{*,t}\right)-u\left(\bm{\mu}^t,J\right) \leq   3BLD^{\frac{\alpha}{2}}T^{-\alpha \gamma}+At^{\theta}
	\end{align}  
	Therefore, the regret can be bounded by 
	\begin{align}
	\mathbb{E}[R_n(T)] \leq & \sum_{t=1}^{T}\left(3BLD^{\frac{\alpha}{2}}T^{-\alpha \gamma}+At^{\theta}\right) \\
	\leq & 3BLD^{\frac{\alpha}{2}}T^{1-\alpha \gamma}+\dfrac{A}{1+\theta}T^{1+\theta}.
	\end{align}
\end{proof}
The over all regret is now bounded by applying the above Lemmas. 

\begin{proof}[Proof of Proposition \ref{theo:upper_bound}]
	First, let $K(t)=t^{z}\log(t)$ and $h_{T}=\lceil T^{\gamma} \rceil$, where $0<z<1$ and $0<\gamma<\frac{1}{D}$; let $H(t) = Bt^{-z/2}$; let the condition $2H(t)+2BLD^{\frac{\alpha}{2}}T^{-\alpha \gamma} \leq At^\theta$ be satisfied for all $1<t<T$. Combining the results of above Lemmas, the regret $R(T)$ is bounded by
	\begin{align*}
	R(T) \leq & (1+\eta B)2^D\left(T^{z+\gamma D}\log(T) + T^{\gamma D}\right) \\ 
	& + (1+\eta B)\left(\sum_{k=1}^B{V^{\max} \choose k}\right) \frac{\pi^2}{3} \\
	& + 3BLD^{\frac{\alpha}{2}}T^{1-\alpha \gamma}+\dfrac{A}{1+\theta}T^{1+\theta}
	\end{align*}
	The summands contribute to the regret with leading orders $O(\log(T)T^{z+\gamma D})$, $O(T^{1-\gamma\alpha})$ and $O(T^{1+\theta})$. In order to balance the leading orders, we select the parameters $z, \gamma, A, \theta$ as following values $z=\frac{2\alpha}{3\alpha+D}\in (0,1), \gamma=\frac{z}{2\alpha}\in(0,\frac{1}{D}), \theta=-\frac{z}{2}$, and $A=2B + 2BLD^{\alpha/2}$. Note that the condition \eqref{eq:condition_E3} is satisfied with these values. The the regret $R(T)$ reduces to
	\begin{align*}
	R(T) \leq & (1+\eta B)2^{D}\left(\log(T)T^{\frac{2\alpha+D}{3\alpha+D}} + T^{\frac{D}{3\alpha+D}}\right)  \\
	& +(1+\eta B)\left(\sum_{k=1}^B{V^{\max} \choose k}\right) \frac{\pi^2}{3} \\
	& + \left(3BLD^{\alpha/2}+\frac{2B+2BLD^{\alpha/2}}{(2\alpha+D)/(3\alpha+D)} \right) T^{\frac{2\alpha+D}{3\alpha+D}}
	\end{align*}
	Then the leading order is $O\left( (1+\eta B)2^D  T^{\frac{2\alpha+D}{3\alpha+D}}\log(T)\right)$.
\end{proof}

\section{Proof of Proposition \ref{theo:upper_bound_DF}}\label{proof:theo:upper_bound_DF}
\begin{proof}
	We only consider the regret incurred by the mis-exploitation. In the mis-exploitation phase, we will at least one hypercube $p$ such that $M^{t}(p) < K(t) < C^{t}(p)$. Besides the two counters $M^{t}(p)$ and $C^{t}(p)$, the algorithm also keep a timespan $d^t(p)$, such that after  $d^t(p)$ the observed qualities for hypercube $p$ is larger than $K(t)$.
	
	Since the minimum increment for counters $C^t(p)$ and $M^t(p)$ is 1, we split the task sequence $\{1,2,\dots,T\}$ into $\lceil K(T)\rceil$ segments $\{[t^\prime_i,t^\prime_{i+1})\}_{i=1}^{\lceil K(T)\rceil}, i\in\mathbb{N}^+$. Such that for any $t \in [t^\prime_i,t^\prime_{i+1})$, we have $\lceil K(t) \rceil = i$. 
	Now let us consider the task $t^\prime_i$ and assume the algorithm enters the mis-exploitation for task $t^\prime_i$, i.e. $M^{t^\prime_i}(p) < K(t^\prime_i) \leq i < C^{t^\prime_i}(p)$. Since the each task has a deadline and the maximum feedback delay of the replications for task $t$ is $L^t$, one can easily verify that the largest value that $d^t(p)$ can have is $L^{\max}$ which is the maximum deadline of tasks. Therefore, in the worst case, the tasks arriving after $t^\prime_i$ within $L^{\max}$ will enter the mis-exploitation. Let $\lambda$ be the task arrival rate, then the expected regret incurred by mis-exploitation for the $i$-th segment is $\lambda L^{\max} (1+\eta B)$ where $(1+\eta B)$ is the maximum regret for one task. Notice that if $\lambda L^{\max} > t^\prime_{i+1} - t^\prime_{i}, \forall i$, then the upper bound of mis-exploitation regret becomes $(1+\eta B)T$ which grows linearly with the task number $T$. However, based on the design of the control function $K(t)$, we know that when $i\to\infty$ the number of tasks in one segment $t^\prime_{i+1} - t^\prime_{i} \to \infty$. Therefore we must have $\lambda L^{\max} < t^\prime_{i+1} - t^\prime_{i}$ when $i$ is larger than a certain number $i^\prime$. Since there are a total of $\lceil K(T) \rceil$ segments and $K(t) = t^{\frac{2\alpha}{3\alpha + D}}\log(t)$, the regret of mis-exploitation is 
	\begin{align*}
	R_m (T) &=  \lambda L^{\max} (1+\eta B)\lceil K(T) \rceil \leq \lambda L^{\max} (1+\eta B)(T^{\frac{2\alpha}{3\alpha + D}}\log(T) + 1)
	\end{align*}
\end{proof} 
\end{document}